\documentclass[11pt,letterpaper,oneside]{article}
\usepackage{graphicx}
\usepackage{amssymb,amsmath,amsfonts,amsthm}
\usepackage{bm}
\usepackage{graphicx}
\usepackage{bbm}
\usepackage{epstopdf}
\usepackage{natbib}
\usepackage{amsfonts}
\usepackage{booktabs} % for much better looking tables
\usepackage{array} % for better arrays (eg matrices) in maths
\usepackage[above, below]{placeins}
\usepackage{glossaries}
\usepackage{verbatim}
\usepackage{setspace}

\usepackage[a4paper,left=1in,top=1in,right=1in,bottom=1in,ignoreheadfoot]{geometry}
\usepackage{algorithm}
\usepackage{setspace}

\providecommand{\A}{\mathcal{A}}

\providecommand{\1}{\mathbbm{1}}

\providecommand{\E}{\mathbb{E}}
\providecommand{\KL}{\mathrm{KL}}

\newcommand\diff{\mathrm{d}}

\usepackage{threeparttable}
\usepackage{pgf}
\usepackage{tikz}
\usetikzlibrary{fit,positioning,arrows,automata}
\tikzset{
    >=stealth',
    punkt/.style={
           rectangle,
           rounded corners,
           draw=black, very thick,
           text width=6.5em,
           minimum height=2em,
           text centered},
    pil/.style={
           ->,
           thick,
           shorten <=2pt,
           shorten >=2pt,}
}

\usepackage{caption}
\usepackage{subcaption}
\usetikzlibrary{calc}

\newtheorem{theorem}{Theorem}[section]

\newtheorem{definition}{Definition}

    \def\independenT#1#2{\mathrel{\setbox0\hbox{$#1#2$}%
    \copy0\kern-\wd0\mkern4mu\box0}}

\begin{document}
%\graphicspath{ {/Users/jameswatson/Dropbox/SensitiveDecisions/Paper/StatScienceversion/Plots/} }
\graphicspath{{Plots/}} %     
%\graphicspath{ {C:/Users/chrisholmes/Dropbox/SensitiveDecisions/Paper/StatScienceversion/Plots/} }

\setlength{\parindent}{0pt}
\setlength{\parskip}{\baselineskip}

\title{Approximate Models and Robust Decisions}

\author{James Watson and Chris Holmes\footnote{Chris Holmes is Professor of Statistics, Department of Statistics, University of Oxford, UK. email: cholmes@stats.ox.ac.uk}}
\maketitle

\begin{abstract}
Decisions based partly or solely on predictions from probabilistic models may be sensitive to model misspecification. Statisticians are taught from an early stage that ``all models are wrong'', but little formal guidance exists on how to assess the impact of model approximation on decision making, or how to proceed when optimal actions appear sensitive to model fidelity. 
This article presents an overview of recent developments across different disciplines to address this. We review diagnostic techniques, including graphical approaches and summary statistics, to help highlight decisions made through minimised expected loss that are sensitive to model misspecification. 
We then consider formal methods for decision making under model misspecification by quantifying stability of optimal actions to perturbations to the model within a neighbourhood of model space. 
This neighbourhood is defined in either one of two ways. Firstly, in a strong sense via an information (Kullback-Leibler) divergence around the approximating model. Or using a nonparametric model extension, again centred at the approximating model, in order to `average out' over possible misspecifications.
This is presented in the context of recent work in the robust control, macroeconomics and financial mathematics literature. We adopt a Bayesian approach throughout although the methods are agnostic to this position.

\vspace{0.2in} \noindent {\bf{Keywords}}: Decision theory; Model misspecification; {\em{D-open}} problems; Kullback-Leibler divergence; Robustness; Bayesian nonparametrics

\end{abstract}

\section{Introduction}

This article presents recent developments in robust decision making from approximate statistical models. The central message of the paper is this, that the consequence of statistical model misspecification is contextual and hence should be dealt with under a decision theoretic framework \citep{berger1985statistical, parmigiani2009decision}. As a trivial illustration consider the following example: suppose that data arise from an exponential distribution, $x \sim \exp(\lambda)$, yet the statistician adopts a normal model, incorrectly assuming  $x \sim N(\mu, \sigma^2)$. If interest is in the estimation of the mean $E[X]$ and the sample size is large there may be little consequence in the misspecification. However if the focus is on the probability of an interval event, say $X \in [a, b]$, then there might be far reaching consequences in using the model. %Moreover, if we were seeking to estimate the mean, yet losses were critically sensitive to tail errors in mean specification then the normal model might no longer be adequate. 
Of course this is a toy problem and careful model checking and refinement will help in reducing misspecification, but pragmatically, especially in modern high-dimensional data settings, it seems to us inappropriate to separate the issue of model misspecification from the consequences, context, and rationale of the modelling exercise. 

Statisticians are taught from an early stage that ``essentially, all models are wrong, but some are useful'' \citep*{box1987empirical}. By ``wrong'' we will take to mean misspecified and by ``useful'' we will take to mean helpful for aiding actions (taking decisions). We will refer to such situations as {\em{D-open}} problems, to highlight that Nature's true model is outside of the decision makers knowledge domain, c.f. {\em{M-open}} in Bayesian statistics which refers to problems in updating beliefs when the model is known to be wrong \citep*{bernardo1994bayesian}. We will adopt a Bayesian standpoint throughout although the approach we develop is generic. We will assume there is uncertainty in some aspect of the world\footnote{\cite{savage1954foundations} refers to $\Theta$ as the ``small world'' relevant to the decision.}, $\theta \in \Theta$, which if known would determine the loss in taking an action $a$ as quantified through a real-valued measurable loss-function, $L_a(\theta)$. The loss will often be a joint function of states and observables, $L_a(\theta, x)$, although we shall suppress this notation for convenience. Uncertainty in $\theta$ is characterised via a probability distribution $\pi_I(\theta)$ given all available information $I$. Without loss of generality we will assume that $\theta$ relates to parameters of a probability model and information $I$ is in the form of data, $x \in {\cal{X}}$, and a joint model $\pi(x, \theta)$, such that,
$$
\pi_I(\theta) \equiv \pi(\theta | x) \propto f(x; \theta) \pi(\theta),
$$
where $\pi(x, \theta)$ is factorised according to the sampling distribution (or likelihood) $f(x; \theta)$ and the prior $\pi(\theta)$; although more generally $\pi_I(\theta)$ simply represents the statisticians best current beliefs about the value of the unknown state $\theta$. Following the axiomatic approach of \cite{savage1954foundations} the rational coherent approach to decision making is to select an action $\hat{a}$ from the set of available actions $a \in A$ so as to minimise expected loss,
\begin{equation}
\label{eq:savage}
\hat{a} = \arg \inf_{a \in A} \E_{\pi_I(\theta)}[L_a(\theta)].
\end{equation}
This underpins the foundations of Bayesian statistics \citep*{bernardo1994bayesian}. The problem is that (\ref{eq:savage}) assumes perfect precision in specifying $\pi(x, \theta)$. In reality the model $\pi(x, \theta)$ is misspecified, such that the decision maker acknowledges that $f(x; \theta)$ may not be Nature's true sampling distribution or $\pi(\theta)$ does not reflect all aspects of prior subjective beliefs in $f(x; \theta)$ or on the marginal $\pi(x) = \int_{\theta} \pi(x, \theta) d\theta$. This paper presents diagnostics and formal methods to assist in exploring the potential impact of this misspecification.
 
It is important to note that we will not spend much time on the area of pure inference problems such as robust estimation of summary functionals for which there is a substantial literature \citep*{huber2011robust}, or on recent work on the use of loss functions to construct posterior models \citep*{Bissiri2012a,Bissiri2013}. We shall also pass quickly over the use of conventional prior sensitivity analysis and robust ``heavy tailed'' priors and likelihoods.  We are principally concerned with \textit{ex-post}\footnote{Meaning here 'once the modelling has been completed'. In a Bayesian setting, this refers to dealing directly with the posterior quantities.} settings where $\pi(x, \theta)$ has been specified to the best of the modellers ability under the practical constraints of computation and time, and where concerns arise as to whether $\pi(x, \theta)$ represents the modeller's true marginal $\pi(x)$ to sufficient precision. This is particularly important when $\theta$ pertains to a high-dimensional complex model or to the value of a future predicted observation.

There is a rich literature in Bayesian statistics on model robustness, the vast majority of which relates to sensitivity to specification of the prior $\pi(\theta)$. We review the material in detail below but mention here the overviews in \cite{berger1994overview}, \cite{insua2000robust} and \cite{ruggeri2005robust}. Bayesian robustness was a highly active area through the 1980s to mid-1990s. Interest tailored off somewhat since that time, principally due to the arrival of computational methods such as Markov chain Monte Carlo (MCMC) coupled with developments in hierarchical models, nonlinear models and nonparametric priors, see e.g.  \cite{chipman1998bayesian}, \cite{robert2004monte}, \cite{rasmussen2006gaussian},  \cite{denison2002nonlinear}, and \cite{hjort2010bayesian}. These methods allow for very flexible model specifications alleviating the historic concern that $\pi(x, \theta)$ was indexing a restrictive sub-class of models. 
% Moreover, there was perhaps a lack of quantitative robust procedures with well characterised theoretical properties, in contrast to the development of M-estimators and subsequent robust estimators in the non-Bayesian literature. 
However, a number of recent factors merit a reappraisal. In the 1990s and 2000s computational advances and hierarchical models broadly outpaced the complexity of data sets being considered by statisticians. In more recent times very high-dimensional data are becoming common, the so called ``big data'' era, whose size and complexities prohibit application of fully specified carefully crafted models, (e.g. \citep*{national2013Frontiers}, Chapter 7). Related to this, approximate probabilistic inference techniques that are misspecified by design have emerged as important tools for applied statisticians tackling complex inference problems. For example, models involving composite likelihoods, integrated nested Laplace approximations (INLA), Variational Bayes, Approximate Bayesian Computation (ABC), all start with the premise of misspecification, see e.g. \cite{beaumont2002approximate}, \cite{fearnhead2012constructing}, \cite{marjoram2003markov}, \cite{marin2012approximate},  \cite{minka2001expectation}, \cite{ratmann2009model}, \cite{rue2009approximate}, \cite{varin2011overview}, and \cite{wainwright2003}. Finally there have been recent developments in coherent risk measures within the macroeconomics and mathematical finance literature, building on areas of robust control, which are of importance and relevance to statisticians, as outlined in Section 2 below.

The rest is as follows. In Section 2 we review some background literature on decision robustness and quantification of expected loss under model misspecification. In Section 3 we review diagnostic tools to assist applied statisticians in identifying actions which may be sensitive to model fidelity. Section 4 presents formal methods for summarising decision stability, by exploring the consequence of misspecification within local neighbourhoods around the approximate model. Section 5 contains illustrations. Conclusions are made in Section 6.

\section{Background on decisions under model misspecification}
We first review some of the background literature on decisions made under model misspecification.
\subsection{Minimax}

The first axiomatic approach to robust statistical decision making was made by Wald (\citeyear{wald1950statistical}). In the absence of a true model, Wald interpreted the 
decision problem as a zero sum two-person game, following Von Neumann and Morgenstern's work on game theory \citep*{von1947theory}. To be 
robust the statistician protects himself against the worst possible outcome, selecting an action $\hat{a}$ according to the minimax rule, which for the purposes of this paper we can consider as\footnote{Wald's original work considered selection of decision functions, $\delta(x) \in A$, by non-conditional loss quantified as frequentist risk, $R[F_X, \delta(x)] = \int L(\delta, x) F(dx),$ with $x \in {\cal{X}}$ from unknown distribution $F_X$.}, 
$$
\hat{a} = \arg \inf_{a \in A} \left[ \sup_{\theta \in \Theta} L_a(\theta) \right].
$$
This is akin to the decision maker playing a two-person game with a malevolent Nature, where losses made by one agent will be gained by the other (zero sum). On selection of an action, Nature will select the worst possible outcome, equivalent to the assumption of a point mass distribution taken reactively to your choice of action, $$ \delta_{\theta^*_a}(\theta),
$$
where,
$$
\theta^*_a = \arg \sup_{\theta \in \Theta} L_a(\theta).
$$
Although elegant in its derivation the minimax rule has severe problems from an applied perspective. The decision maker following the minimax rule is not rational and treats all situations with extreme pessimism. It assumes that Nature is reactive in selecting $\delta_{\theta^*_a}(\theta)$ for your choice of $a \in A$ irrespective of the evidence from existing information $I$ on the plausible values of $\theta$. Subsequent to Wald there has been considerable work to develop more applied procedures that protect against less extreme outcomes.

\subsection{Robust Bayesian statistics}

Under a strict Bayesian position there is no issue with model robustness. You precisely specify your subjective beliefs through $\pi(x, \theta)$ and condition on data to obtain posterior beliefs, taking actions according to the Savage axioms. However, even the modern founders of Bayesian statistics acknowledged issues with an approach that assumes infinite subjective precision,  
\begin{itemize}
\item[]  ``Subjectivists should feel obligated to recognise that any opinion (so much more the initial one) is only vaguely acceptable... So it is important not only to know the exact answer for an exactly specified initial problem, but what happens changing in a reasonable neighbourhood the assumed initial opinion.'' De Finetti, as quoted in \cite{dempster1975subjectivist}
\item[] ``...in practice the theory of personal probability is supposed to be an idealization of one's own standard of behaviour; that the idealization is often imperfect in such a way that an aura of vagueness is attached to many judgements of personal probability...'' \cite{savage1954foundations}
\end{itemize}
As Berger points out, many people somewhat distrust the Bayesian approach as ``Prior distributions can never be quantified or elicited exactly (i.e. without error), especially in finite amount of time'' --  Assumption II in \cite{berger1984robust}. In which case what does the resulting posterior distribution $\pi(\theta | x)$ actually represent? 

An intuitive solution is to first specify an operational model $\pi_0$, to the best of your available time and ability,  and then investigate sensitivity of inference or decisions to departures 
around $\pi_0$, typically assuming that $f(x ; \theta)$ is known so that divergence is with respect to the prior. This idea has origins in the work of \cite{robbins1952asymptotically} and \cite{good1952rational} with many important contributions since that time. We mention just a few pertinent areas below, referring the interested reader to the review articles of \cite{berger1984robust,berger1994overview},  \cite{wasserman1992recent}, and \cite{ruggeri2005robust}, as well as the collection of papers in the edited volumes of \cite{kadane1984robustness} and \cite{insua2000robust}. 

The resulting robust Bayesian methods are usefully classified as either ``local'' or ``global''.
Local approaches look at functional derivatives of posterior quantities of interest with respect to perturbations around the baseline model, e.g. \cite{ruggeri1993infinitesimal} \cite{sivaganesan2000global}; see also \cite{kadane1978stable} who consider asymptotic stability of decision risk. 
Global approaches consider variation in a posterior functional of interest, $\psi = \int h(\theta) \pi(\theta | x) d \theta$, within a neighbourhood $\pi \in \Gamma$ centred around the prior model $\pi_0$.  A typical quantity would be the range $(\psi^{\inf}, \psi^{\sup})$
%$(\underbar{{\psi}})$ $(\overline{\psi})$
 where $\psi^{\inf} = \inf_{\pi \in \Gamma} \psi$ and $\psi^{\sup} = \sup_{\pi \in \Gamma} \psi$.
The challenge is to define the nature and size of $\Gamma$ so as to capture plausible ambiguity in $\pi_0$, while taking into account factors such as ease of specification and computational tractability, Berger (\citeyear{berger1994overview}; \citeyear{berger1985statistical} section 4.7). 
%The two most popular being $\epsilon$-contamination and $\Gamma$-minimax; see Berger (\citeyear{berger1985statistical}, section 4.7). Let $\psi$ denote the posterior functional of interest, $\psi = \int h(\theta) \pi(\theta | x) d \theta$, then a typical quantity computed with global robustness would be the range $(\psi_{\inf}, \psi_{\sup})$
%$(\underbar{{\psi}})$ $(\overline{\psi})$
% where $\psi_{\inf} = \inf_{\pi \in \Gamma} \psi$ and $\psi_{\sup} = \sup_{\pi \in \Gamma} \psi$.
One important example is the  {\em{$\epsilon$-contamination}} neighbourhood \citep{berger1986robust} formed by the mixture model, 
$$
\Gamma = \{\pi = (1 - \epsilon) \pi_0 + \epsilon q, q \in \cal{Q}\} ,
$$
where $\epsilon$ is the perceived contamination error in $\pi_0$ and $\cal{Q}$ is a class of contaminant distributions. It is usual to restrict $\cal{Q}$ so that it is not ``too big'', for instance by including only uni-modal distributions \cite{berger1994overview}, for which it is shown that the solutions have tractable form. Other approaches consider frequentist risk, such as {\em{$\Gamma$-minimax}} that investigates the minimax Bayes (frequentist) risk of $\psi^{\sup}$ for $\pi \in \Gamma$ whereas {\em{conditional  $\Gamma$-minimax}} procedures \citep*{vidakovic2000gamma} study the maximum expected loss across prior distributions within $\Gamma$, this being perhaps closest to the approach we develop here.

One distinction between these approaches and this paper, is that we shall be concerned with robustness to misspecification on only those states $\theta$ that enter into the loss function $L_a(\theta)$. This facilitates application to high-dimensional problems for which specification of $\Gamma$ may be difficult \citep{Sivadiscussionberger} and helps tackle the thorny issue that changing the likelihood changes the interpretation of the prior  (\cite{ruggeri2005robust}, page 635). 

\subsection{Robust control, macroeconomics and finance}

Independent of the above developments in statistics, control theorists were investigating robustness to modelling assumptions. Control theory broadly concerns optimal intervention strategies (actions) on stochastic systems so as to maintain the process within a stable regime. Hence it is not surprising that decision stability is an important issue. When the system is linear with additive normal (white) noise the optimal intervention is well known \cite{whittle1990risk}. Robust control theory, principally developed by Whittle, considers the case when Nature is acting against the operator through stochastic buffering by non-independent noise, see  \cite{whittle1990risk}. Whittle established that under a malevolent Nature with a bounded variance an optimal intervention can be calculated using standard recursive algorithms. 

In Economics one early criticism of the Savage axioms was that the framework could not distinguish between different types of uncertainty. \cite{gilboa1989maxmin} developed a theory of maxmin Expected Utility in part to counter the famous Ellsberg paradox\footnote{The standard setting for the Ellsberg paradox is as follows: imagine two urns each containing 100 balls and every ball is either red or blue. One is told that the first urn (A) has 50 red balls and 50 blue balls exactly. No more information is given about the second urn (B). Suppose you win 100\$ if you pick a red ball, which urn would you choose? So there exists a set of alternatives which are equal in expected value (under any reasonable prior) but which appear to have different empirical preferences.}which extends standard Bayesian inference to a setting with multiple priors in the form of a closed convex set $\Gamma$. An action is then scored by its expected loss under the least favourable prior within that set. Their \citeyear{gilboa1989maxmin} paper formalises this and provides a solution to the Ellsberg paradox. When $\Gamma$ contains only one prior, we are back again in the usual Bayesian setting. The set $\Gamma$ can be seen as describing the decision-maker's aversion to uncertainty. This work is closely related to $\Gamma$\textit{-minimax} (for which the Ellsberg paradox is also used as a motivating example, see section 1 of \cite{vidakovic2000gamma}).

Again working in economics, Hansen and Sargent in a series of influential papers (e.g. \citeyear{hansen2001acknowledging}, \citeyear{hansen2001robust}), generalised ideas from \cite{whittle1990risk} and \cite{gilboa1989maxmin} motivated by problems in macroeconomic time series. They define a robust action as a local-minimax act within a Kullback-Leibler (KL) neighbourhood of $\pi_I(\theta)$ through exploration of, 
$$
\psi_{(a)}^{\sup}(C) = \sup_{\pi \in \Gamma_C} \E_{\pi}[L_a(\theta)]
$$
where $\Gamma_C$ denotes a KL ball of radius $C$ around $\pi_I$, 
$$
\Gamma_C = \{ \pi : \int \pi(\theta) \log \left( \frac{\pi(\theta)}{\pi_I(\theta)} \right) d\theta \le C \}.
$$
We will use $\pi^{\sup}_{a,C}$ to denote the corresponding local-minimax distribution, 
$$
\pi^{\sup}_{a,C} = \arg \sup_{\pi \in \Gamma_C} \E_{\pi}[L_a(\theta)].
$$
Figure \ref{GraphicalKLMinimax} shows a pictorial representation of this constrained minimax rule, where the reference distribution $\pi_I$ is a point in the space $\mathcal{M}$ of all distributions on $\theta$ (represented by the rectangle) and the least favourable distribution $\pi_{a,C}^{\sup}$ is contained within the neighbourhood $\Gamma_{C}$ (represented by the circle of radius $C$). The Wald minimax distribution is given by $\delta_{\theta_a^*}(\theta)$. Hansen and Sargent showed how $\pi_{a,C}^{\sup}$ and $\psi_{(a)}^{\sup}$ can be computed for dynamic linear systems with normal noise, see \cite{hansen2008robustness} for a thorough review and references. 

 \begin{figure}
\centering
\begin{tikzpicture}[scale=4.,cap=round,>=latex]
\def\x{0.65} % this defines the size of the circle 
	\draw (-1.5,-1.05) rectangle (1.5,1.05) ;
	\draw[thick] (0cm,0cm) circle(\x cm);
           \draw[gray] (0cm,0cm) -- (30:\x cm);
              % dots at each point
           \filldraw[black] (30:\x cm) circle(0.4pt);
           \filldraw[black] (45:1.4 cm) circle(0.4pt);
	\draw (42:1.27 cm) node[fill=white] {$\delta_{\theta_a^*}(\theta)$};
	\filldraw[black] (30:0cm) circle(0.4pt);
	\draw (30:.6*\x cm) node[fill=white] {$C$};
	\draw (326:1.6cm) node[fill=white] {$\mathcal{M}$};
	\draw (30:1.25*\x cm) node[fill=white] {$\pi^{\sup}_{a,C}$};
	\draw (110:0.1) node[fill=white] {$\pi_I$};
	\draw (330:0.815*\x cm) node[fill=white] {$\Gamma_{C}$};
\end{tikzpicture}
\caption{Graphical representation of local-minimax model $\pi^{\sup}_{a,C}$ within a Kullback-Leibler ball of radius $C$ around the reference model $\pi_I$, with global (Wald's) minimax density $\delta_{\theta_a^*}(\theta)$.}
\label{GraphicalKLMinimax}
\end{figure}
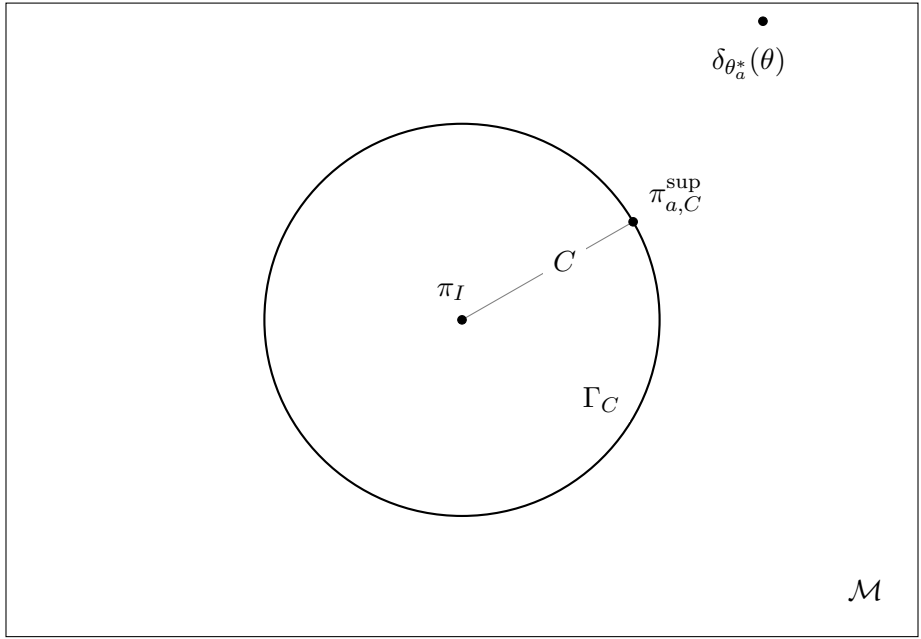

\citeauthor{breuer2013systematic} (\citeyear{breuer2013systematic}, \citeyear{breuer2013measuring}), building on the work of Hansen and Sargent, derived corresponding results for arbitrary probability measures $\pi_I(\theta)$. Under mild regularity conditions, and using results from exponential families and large deviation theory they obtain the exact form of $\pi^{\sup}_{a,C}$ for any $\pi_I(\theta)$ given the KL ball of size $C$, as well as an estimate for $\psi_{(a)}^{\sup}$, see also \citeauthor{ahmadi2012entropic} (\citeyear{ahmadi2011information}, \citeyear{ahmadi2012entropic}). In Section 4 we  derive the same result using an alternative, less general, but perhaps more intuitive proof. Before considering these formal methods we shall start with exploratory diagnostics and visualisation methods.

\section{{\em{D-open}} diagnostics}\label{diagnostics}

All good statistical data analysis begins with graphical exploration of information and evaluation of summary statistics before formal modelling takes place. In this section we consider some graphical displays to aid understanding of when and where actions are sensitive to modelling assumptions.  Section \ref{applications} further illustrates these ideas. Despite the importance of graphical statistics, there are few if any established tools for investigation of decision stability, in contrast to the multitude of methods for investigating model discrepancy and misspecification \citep*{belsley2005regression,gelman2007data,bayesvisualisation}.
Here we consider three graphical displays that concentrate on the relationship between the loss function $L(a,\theta)$ and `model' or posterior $\pi_I(\theta)$ for a given $a$. These are examples of how a set of available actions $a\in\A$ can be graphically compared. They could be displayed as a preliminary step to a formal analysis of sensitivity.

\subsection{Value at Risk (Quantile-Loss)}\label{VAR}

A primary tool for assessing the sensitivity with respect to misspecification of a functional of interest, is the distribution of loss, where $Z_a$ denotes the random loss variable under $\pi_I(\theta)$:
$$
F_{Z_a}(z) =  Pr(Z_a \le z) = \int_{\theta \in \Theta} I[L_a(\theta) \le z] \pi_I(d \theta) ,
$$
where $I[\cdot]$ is the identity function. We use notation $f_a(z) = F_{Z_a}(d z)$ to denote the corresponding density function and $F_{Z_a}^{-1}(q)$, for $q \in [0, 1]$, is the inverse cumulative distribution or quantile function. For a given $q\in [0,1]$ it is possible to characterise the value of an action $a$ by its quantile loss or \textit{Value at Risk} (VaR; terminology used in finance) $F_{Z_a}^{-1}(1-q)$. \cite{rostek2010} developed an axiomatic framework in which decision-makers can be uniquely characterised by a quantile $1- q$ and rational behaviour (optimal action) is defined as choosing $\hat{a} := \arg\min_{a\in\mathcal{A}} F_{Z_a}^{-1}(1-q)$. Note that this incorporates the minimax rule by taking $q=0$. The author argues that quantile maximisation is attractive to practitioners as its key characteristic is robustness, specifically to misspecification in the tails of the loss distribution. Although single quantiles discard much information contained in $[\pi_I(\theta),L_a(\theta)]$, plotting this function allows for immediate visualisation of how much of the tails are taken up by high loss (low utility) events. With a bag of samples simulated from the posterior marginal $\theta_1,..,\theta_m \sim \pi_I(\theta)$, this is easily approximated:
\begin{enumerate}
 \item Sort the realised loss values, $z_i^{(a)} = L_a(\theta_i)$, from highest to lowest $\{z^{(a)}_{\upsilon_a(1)} \ge z^{(a)}_{\upsilon_a(2)} \ge \ldots \ge z^{(a)}_{\upsilon(m)}\}$, where $\upsilon_a(\cdot)$ defines the sort mapping.
\item For $q\in[0,1]$, approximate $F_{Z_a}^{-1}(1-q)$ by linear interpolation of the points $\left(x=k/m, y = z^{(a)}_{\upsilon_a(k)}\right)$.
\end{enumerate}

\subsubsection{Coherent diagnostics}

In mathematical finance, summary statistics defined on loss distributions are known as \textit{risk measures}.  
VaR is a particularly controversial risk measure as it is widely used \citep[written into official regulations, see][]{basel:96}
but is not \textit{coherent}\footnote{Note that this is a different definition of coherence from Bayesian coherence, discussed in section \ref{bayescoherence}.} \citep{artzner:99}\footnote{A coherent risk measure $\rho$ has the following properties: \textit{translational invariance}: $\rho(Z+c) = \rho(Z) + c$, where $c$ is a constant; \textit{monotonicity}: if $Z$is stochastically dominated by $Y$, then $\rho(Z)\leq\rho(Y)$; \textit{positive homogeneity}: $\rho(\lambda Z) = \lambda\rho(Z)$, for $\lambda\geq 0$; \textit{subadditivity}: $\rho(Z+Y) \leq \rho(Z)+\rho(Y)$. By $\rho(Z+Y)$ we mean the risk measure on the combined loss distribution of the combination of the two actions corresponding to the loss distributions $Z$ and $Y$.}, violating subadditivity\footnote{Subadditivity corresponds to investors decreasing risk by diversifying portfolios.}. This has motivated the use of different diagnostics such as CVaR (see next section). We note that expected loss and minimax are both coherent diagnostics. However, the use of coherence in Bayesian decision theory does not seem appropriate in general as the subadditivity axiom does not hold in many applications\footnote{A trivial example is the following: Consider the two gambles, $A$ lose $\pounds10$ if a fair coin falls on Heads; $B$ lose  $\pounds10$ if a fair coin falls on Tails. For the same coin, if both gambles are taken then one loses  $\pounds10$ with probability 1.}.

\subsection{Conditional value at risk (upper-trimmed mean utility)}\label{CVAR}

\textit{Conditional Value at Risk} \citep*[CVaR,][]{rockafellar2000optimization} is another popular alternative to expected loss (or utility), initially motivated by concerns of coherence. To a statistician it represents the lower trimmed mean of loss (or upper trimmed mean of utility), 
$$
{G}_a(q) = \frac{1}{q} \int_{F_{Z_a}^{-1}(1-q)}^{\infty} z f_a(z) d z .
$$
This gives the expected value of an action conditional on the event ($\theta$) occurring above a quantile of loss (lowest of utility). $q$ can be seen as regulating the amount of pessimism towards Nature, with $\lim_{q\rightarrow0} \sup_a {G}_a(q)$ corresponding to the minimax rule. 

Another strategy for taking in to account model misspecification is by considering the two-sided {\em{trimmed expected loss}}, defined as:
$$
H({q}) = \frac{1}{1- q} \int_{F^{-1}_{Z_a}(q/2)}^{F^{-1}_{Z_a}(1- q/2)} z f_a(z) dz
$$
This is a robust measure of expected loss formed by discarding events with highest and lowest loss. 
Both these statistics are easily approximated using the bag of samples $\{\theta_i\}_{i=1}^n$ and the sort mapping $\upsilon$ defined previously. We use the linear interpolation,
\[ \hat{G}_a(\frac{k}{m}) = \frac{1}{k} \sum_{i=1}^k L_a(\theta_{\upsilon(i)}), ~~~ k=0,..,m\]

For a set of actions $\mathcal{A}$, it is possible to quantify the stability of the optimal action $a^*$ evaluated under expected loss, by observing the first CVaR crossing point. That is to say the first value $q\in[0,1]$ such that $a^*$ is no longer optimal, evaluated under CVaR$(q)$.

\subsection{Cumulative Expected Loss}\label{CER}

The {\em{Cumulative Expected Loss}} (CEL) function for action $a$, defined as, 
$$
J_a(q) = \int_{F_{Z_a}^{-1}(1-q)}^{\infty} z f_a(z) d z .
$$
for  $q \in [0,1]$. The CEL-plot is a monotone decreasing function $J_a(q) $ and an informative graph for highlighting decision sensitivity. The overall shape of $J_a(q)$ provides a qualitative description of decision sensitivity. An action with CEL-plot that is steeply rising as $q \to 0$ is `heavily downside' (see for example figure \ref{BClossdiagnostics} in section \ref{breastcancer}), with expected-loss driven by low-probability high loss outcomes, while CEL-plot rising at 1 indicates `heavy upside'. In particular $J_a(q)$ has a number of useful features:
\begin{itemize}

\item $J_a(q)$ quantifies the contribution to the expected loss of action $a$, from the $ 100  \times (1- q) \%$ set of outcomes with greatest loss. 

\item $J_a(1) = \E_{\pi_I(\theta)}[L_a(\theta)]$, is the expected loss of action $a$, and $\hat{a} = \arg \max_{a \in A} J_a(1)$ is the optimal Savage action.

\item $J_a'(q) =  \inf_{z^* \in\mathbb{R}^+} \left\{ Pr(Z_a \le z^*) = 1 -q \right\}$, 
the gradient of the curve at $J_a(q)$ gives the threshold loss value $z^*$, such that under action $a$ we can expect with probability $(1 -q)$ the outcome to have loss less than or equal to $z^*$.This is the ``value-at-risk'' of action $a$ outlined above, e.g.  \cite{pritsker1997evaluating}.

\item $J_a'(0) = \sup_{\theta \in \Theta} L_a(\theta)$, and hence the Wald minimax action is given by: $\tilde{a} = \arg \min_{a \in A} J_a'(0)$ (the action with steepest gradient as $q \to 0$).
\end{itemize}

\subsection{Sensitivity diagnostics}

If the expected loss estimates appear to be sensitive to the model specification it is useful to know which element of the model is sensitive. Here we propose a simple method and graphical display to estimate the sensitivity with respect to individual data points (likelihood) and/or the prior distribution.
Again, let the model be represented by a bag of Monte Carlo samples $\theta_1,..,\theta_m\sim_{iid}\pi_I$. Thus, in a parametric model, $\pi_I(\theta_i) \propto \Pi_{j=1}^n f(x_j | \theta_i) \pi(\theta_i) $ for data $x_j$, likelihood $f$ and prior $\pi$.

We propose a simple importance sampling method for evaluating $\pi_{I-\{x_j\}}$ and $\pi_{I-\pi}$, denoting respectively the posterior without the datum $x_j$ and the posterior without the prior $\pi$ (the posterior with a flat prior).
The importance weights are given by:
\[ w_i^{-x_j} = \frac{1}{f(x_j | \theta_i)}, \quad w_i^{-\pi} = \frac{1}{\pi(\theta_i)} \]

These weights give leave-one-out (LOO) estimates of the expected loss, where the prior can be considered as one extra data point:
\[ \psi_a^{-x_j} = \frac{1}{\sum_i w_i^{-x_j}} \sum_i w_i^{-x_j} L_a(\theta_i) \]

\[ \psi_a^{-\pi} = \frac{1}{\sum_i w_i^{-\pi}} \sum_i w_i^{-\pi} L_a(\theta_i) \]

Thus the effect of single data points can be evaluated (detection of outliers) as can the effect of the prior, which is especially useful in small sample situations.

We also propose plotting the loss values $L_a(\theta_i)$ against the density estimates $\pi_I(\theta_i)$ (up to a normalising constant). This will highlight situations where the high loss samples are coming from the tails of the distribution $\pi_I$.

\subsection{Motivating synthetic example}\label{syntheticapp}

\begin{figure}
\centering
\includegraphics[scale=.45]{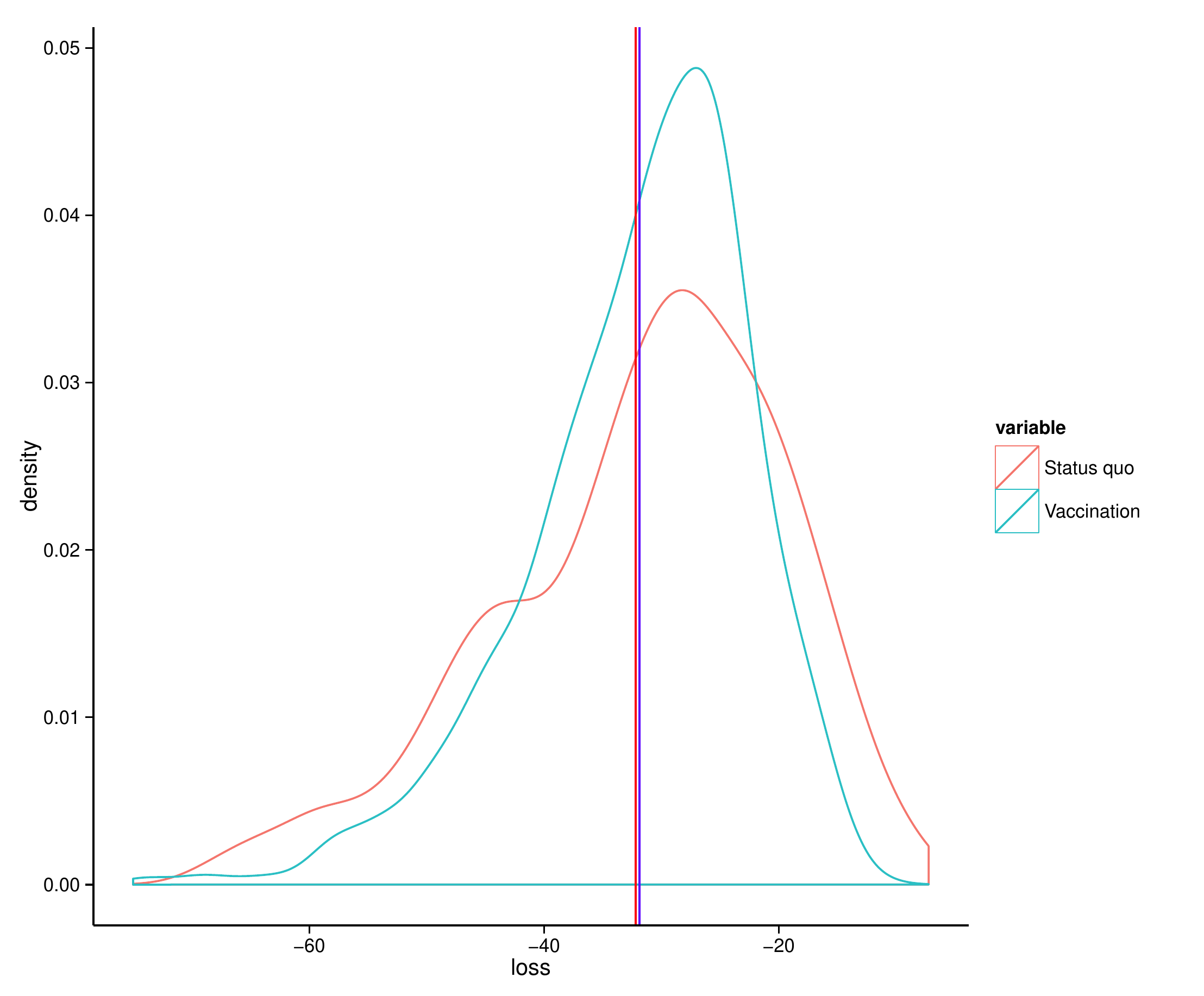}
\caption{Loss distributions of the two actions with mean values shown by vertical lines.}
\label{BCEAlossdists}
\end{figure}
We use an example from the medical decision-making literature to illustrate these three diagnostic plots. Consider a certain infectious disease for which there exists treatment medication and a new vaccine. The problem is whether the vaccination should be publicly funded, or whether the status-quo should remain in place, whereby patients visit their doctor and are prescribed over the counter drugs. This is a standard setting for decisions made by institutions such as NICE\footnote{National Institute for Clinical Excellence.} in the UK, for example. We use a fictitious example of such a decision process taken from \cite{baio2011probabilistic}. The goal is compare the two available actions: widespread vaccination or status quo. The modelling must take into account the uncertainty with regards to the efficacy of the vaccine and its coverage were it to be implemented. With regards the status quo action, the modelling has to consider the number of visits to the GP\footnote{General Practitioner}, complications from the drugs which could lead to extra visits, possible hospitalisation and even death. Each of these outcomes has either a monetary cost (visit to the GP for example) or a utility measured in Quality Adjusted Life Years (QALYs). Therefore to assign a loss value to each action, it is necessary to choose a conversion rate $k$, known as the 'willingness to pay' parameter, exchanging QALYs into pounds. Most of the decision literature focusses on the sensitivity of the decision system with respect to the specification of $k$. The R package BCEA (Bayesian Cost-Effectiveness Analysis) developed by Gianluca Baio implements the model presented in \cite{baio2011probabilistic} and performs a sensitivity analysis around the parameter $k$. The exact details of the model are not of particular interest so we do not expose them here, our main purpose being illustrative.
The model used for this setting has 28 parameters, each with informative prior distributions, these are given in table 1 of \cite{baio2011probabilistic}. MCMC is used to estimate a posterior distribution, all the relevant details can be found in the package documentation, such as the cost function used etc. We ran the model given in BCEA using the default settings. 
We note that all the graphs were produced using our package \textit{decisionSensitivityR} and all the code can be found in its documentation.
In figure \ref{BCEAlossdists} we plot the loss distributions for the two actions with the expected posterior loss values given by the vertical lines. The status quo (red) has lower expected loss but has greater variance than vaccination (blue). The expected loss values are very close together, which would suggest sensitivity to model perturbations.
\begin{figure}
\centering
\includegraphics[scale=.45]{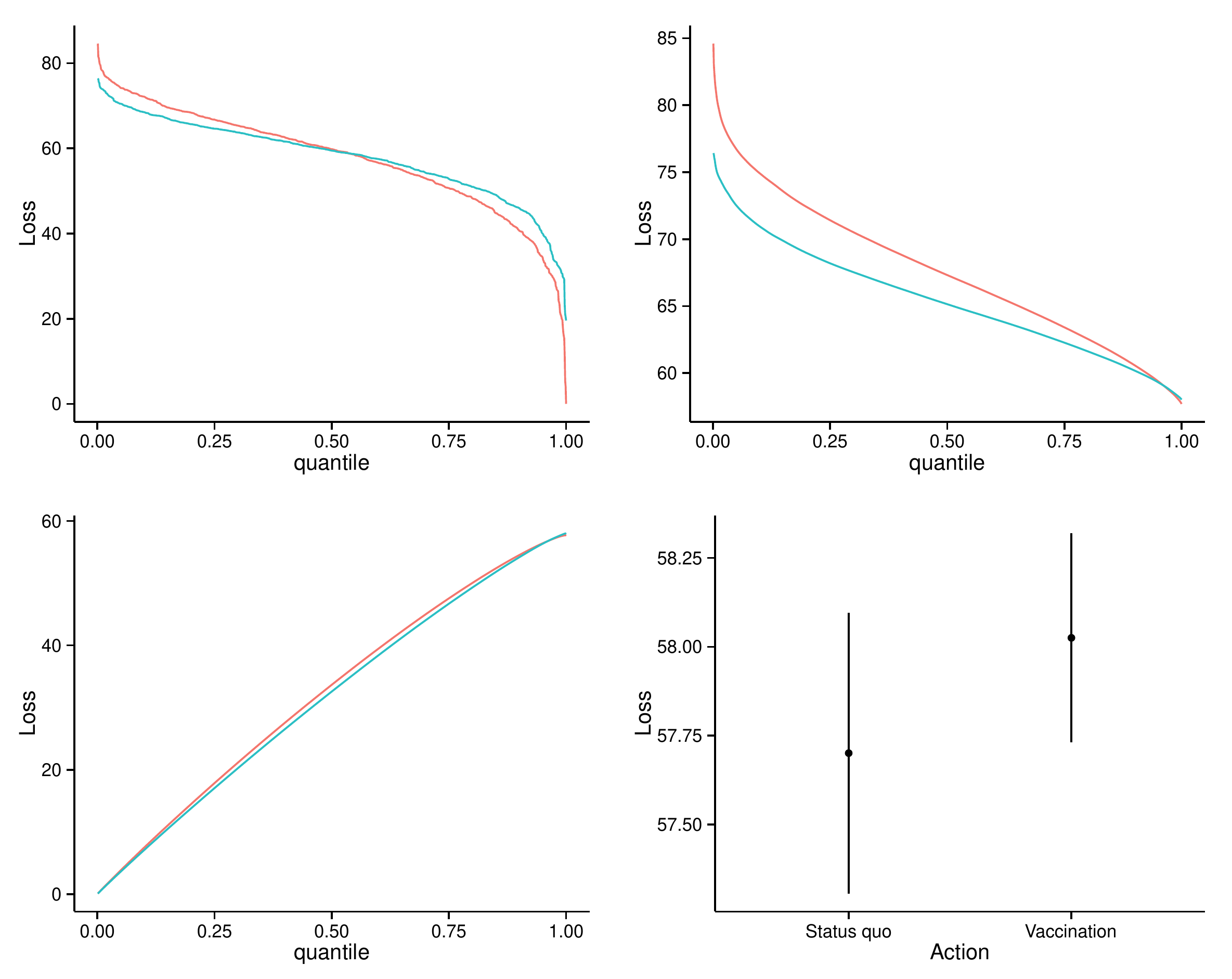}
\caption{Diagnostic plots for the decision system given in \cite{baio2011probabilistic} comparing vaccination (blue) to status quo (red). The 'willingness to pay' parameter set to 21000. From top left to bottom right: Inverse loss distribution; Conditional value at risk; Cumulative expected loss; Expected loss centred at two intervals of standard deviation.}
\label{BCEAdiagnostics}
\end{figure}

Figure \ref{BCEAdiagnostics} illustrates the three diagnostic plots presented above for this application. We see from the inverse loss distribution (top left) that status quo (red) has higher downside loss than vaccination (blue). The CVaR plot (top right) clearer distinguishes the two actions because of this higher downside loss. In this example, the CEL is not informative, but this is context dependent, see figure \ref{BClossdiagnostics} from the breast cancer screening application in section \ref{breastcancer}.

The diagnostic plots and summary statistics presented in this section allow for visualisation of dependencies between the loss function and posterior distribution, which can highlight the impact of model misspecification on decision making.
We now look at perturbations to the model (posterior distribution) in order to measure the sensitivity of the expected loss quantities. 

\section{{\em{D-open}} formal methods}\label{formalmethods}

A natural approach to ex-post model misspecification is via the construction of a neighbourhood of `close' models. This allows for either a study of the variation of the quantity of interest (expected loss) $\psi_{{a}}$ over all models in this neighbourhood, or can guide the construction of a nonparametric extension of the model. In this section, we explore both ideas, each providing the statistician with a different set of tools to estimate the sensitivity of the decision system.

For ease of comprehension, all the notation used throughout this paper is summarised in a glossary in Appendix \ref{app:glossary}. 

\subsection{Kullback-Leibler neighbourhoods}\label{KLanalytic}

To investigate formally the stability of decisions to model misspecification we suggest following an approach taken in \cite{hansen2001robust} and study the variation of expected loss $\psi_{(a)}$ over models within a KL ball $\Gamma$ around the marginal posterior density, $\pi_I(\theta)$, of the approximating model on the states that enter into the loss function. We will assume after linear transformation that the loss can be bounded, a reasonable assumption for almost all applied problems\footnote{In practice it is always possible to cap the loss. For instance, any model by which $\theta$ is simulated using MCMC this assumption is made. In finance, the potential losses incurred by any individual or organisation could be bounded by an arbitrarily large number, say $\pm$GDP of the US.}.

\subsubsection{Properties}\label{KLproperties}

\paragraph{Analytical form of least favourable distribution}
It is well known that the KL divergence is not symmetric, in general $\KL(\pi^*|| \pi) \ne \KL(\pi || \pi^*)$ for $\pi^* \ne \pi$, and following others we consider the neighbourhood $\Gamma_C =\{\pi : \KL (\pi || \pi_I) \le C \}$.
This might be considered the more natural setting as here the KL divergence represents the expected self-information log-loss in using an approximate model $\pi_I$ when Nature is providing outcomes according to the probability law $\pi$. 
The alternative neighbourhood is considered in the Appendix \ref{KLreverse}. However in the Monte Carlo setting where $\pi_I$ is represented by a set of $\{\theta_i\}_{i=1}^m$ each with weight $1/m$, then any distribution $\pi$ that is a reweighing of $\theta_i$'s satisfies: $\KL(\pi_I || \pi) \ge \KL(\pi || \pi_I)$. Because we are looking at the variation of the expected loss $\psi_{(a)}$ across the $\Gamma_C$, we want to the neighbourhood to be more exclusive for fixed values of the radius $C$.

Surprisingly this situation leads to a least favourable distribution solution with a simple form.
\begin{theorem}\label{theorem:1}
\label{minimax}
Let $\pi^{\sup}_{a,C} = \arg\sup_{\pi \in \Gamma_C} \E_{\pi}[L_a(\theta)]$, with $\Gamma_C = \{\pi : \KL(\pi \parallel \pi_I) \le C\}$ for $C \ge 0$.  
Then the solution $\pi^{\sup}_{a,C}$ is unique and has the following form, 
\begin{equation}
\label{eq:pi_sup}
\pi^{\sup}_{a,C} = Z^{-1}_C \pi_I(\theta) \exp[\lambda_a(C) L_a(\theta)]
\end{equation}
where $Z_C = \int  \pi_I(\theta) \exp[\lambda_a(C) L_a(\theta)] d\theta$ is the normalising constant or partition function, for which we assume $Z_C < \infty$, and $\lambda_a(C)$ is a non-negative real valued monotone function.
\end{theorem}

\begin{proof}
The function minimisation problem, $\pi^{\sup}_{a,C} = \arg\max_{\pi \in \Gamma_C} \E_{\pi}[L_a(\theta)]$, has an unconstrained Lagrange dual form, see for example \cite{hansen2006robust} (pages 58-60), 
$$
\pi^{\sup}_{a,C} = \arg\inf_{\pi\in\mathcal{M}}\left\{ \E_{\pi}[-L_a(\theta)] + \eta_a^{-1} \KL(\pi \parallel \pi_I) \right\}
$$
for some $\eta_a =\eta_a(C)$ is a penalisation parameter with $\eta_a \in [0, \infty)$, and is monotone increasing in $C$. Hence, 
\begin{eqnarray}
\pi^{\sup}_{a,C} & = & \arg\inf_{\pi\in\mathcal{M}} \left\{ \int -L_a(\theta) \pi(\theta) d \theta + \eta_a^{-1} \int \pi(\theta) \log \left(\frac{\pi(\theta)}{\pi_I(\theta)}\right) d\theta \right\} \nonumber \\
& = & \arg\inf_{\pi\in\mathcal{M}} \left\{ \int \pi(\theta) \log \left(\frac{\pi(\theta)}{\pi_I(\theta) \exp[\eta_a L_a(\theta)]}\right) d\theta \right\} \nonumber \\
& \propto &  \pi_I(\theta) \exp[\eta_a L_a(\theta)] 
\end{eqnarray}
The uniqueness arises from the convexity of the KL loss. The result follows, taking $\lambda_a(C) = \eta_a$.
\end{proof}

By a similar argument the distribution of minimum expected loss follows:
\[\pi^{\inf}_{a,C} \propto \pi_I(\theta) \exp[-\lambda(C) L_a(\theta)]\]
Note by assuming bounded loss functions we can ensure the integrability of the densities. \cite{breuer2013systematic} and  \cite{ahmadi2012entropic} derive the same result more generally but perhaps less intuitively. \cite{breuer2013systematic} gives more general conditions on when the solution exists.

The $\Gamma_C$ least favourable distributions, $\{\pi^{\inf}_{a,C}, \pi^{\sup}_{a,C}\}$, have an interpretable form as exponentially tilted densities, tilted toward the exponentiated loss function, with weighting $\lambda_a(C)$ a monotone function of the neighbourhood size $C$. For linear loss, $L_a(\theta) = A\theta$, the local least favourable $\pi^{\sup}_{a,C}$ is the well known Esscher Transform used for option pricing in actuarial science. The tilting parameter $\lambda_a(C)$ is a function of the neighbourhood size $C$, but we will write $\lambda_a$ for convenience. $\lambda_a$ and $C$ can be thought of as interchangeable, as there is a bijective mapping between $C\geq 0$ and $\lambda_a\geq 0$, although this is not a linear mapping.

Following Theorem \ref{theorem:1}, the corresponding range of expected losses for each action $(\psi^{\inf}_{(a)}, \psi^{\sup}_{(a)})$ can then be plotted as a function of $C$ for each potential action. Formally we should write 
 $[\psi^{\inf}_{(a)}(C), \psi^{\sup}_{(a)}(C)]$  although for ease of notation we will often suppress $C$ from the expression unless clarity dictates. The constraint $\KL(\pi \parallel \pi_I) \le C$ will result in $\pi_{a,C}^{\sup}$ lodging on the boundary as the expected loss can always be increased by diverging toward $\delta_{\theta^*_a}(\theta)$ for any distribution with $\KL(\pi \parallel \pi_I) < C$. Substituting the solution (\ref{eq:pi_sup}) into the KL divergence function gives,   
\begin{eqnarray*}
\KL(\pi^{\sup}_{a,C} \parallel \pi_I) 
  & = & \int \pi^{\sup}_{a,C}(\theta) \log \left(Z_{\lambda_a}^{-1} \exp[\lambda_a L_a(\theta)] \right)  d\theta \\
    & = & \lambda_a \E_{\pi^{\sup}_{a,C}}[L_a(\theta)]  - \log Z_{\lambda_a}
 \end{eqnarray*}
So, given neighbourhood size $C$, the KL divergence $\KL(\pi^{\sup}_{a,C} \parallel \pi_I)$ is $ \lambda_a(C)$ times the expected loss under $\pi^{\sup}_{a,C}$ minus the log partition function. Moreover, by Jensen's inequality,
\begin{eqnarray*}
\KL(\pi^{\sup}_{a,C} \parallel \pi_I) & = & \lambda_a \E_{\pi^{\sup}_{a,C}}[L_a(\theta)]  - \log \E_{\pi_I}[\exp(\lambda_a L_a(\theta)] \\
& \le & \lambda_a \left[ \E_{\pi^{\sup}_{a,C}}[L_a(\theta)] - \E_{\pi_I}[L_a(\theta)] \right]
 \end{eqnarray*}
The KL divergence is bounded above by $\lambda_a$ times the difference in expected loss between the approximating and the contained minimax models.

By plotting out the interval $[\psi^{\inf}_{(a)}(C), \psi^{\sup}_{(a)}(C)]$ for each action as a function of KL divergence constraint $C : \KL(\pi \parallel \pi_I) \le C$ we can look for crossing points between the supremum loss $\psi^{\sup}_{(\hat{a})}$ under the optimal action $\hat{a}$ chosen by the approximating model and the infimum loss under all other actions, $\psi^{\inf} := \inf_{a \in A \setminus \hat{a}} \{ \psi^{\inf}_{(a)} \}$

\paragraph{Bayesian coherence}\label{bayescoherence}

Adapting results from \citet{Bissiri2013}, we are are able to state the following result regarding the uniqueness of Kullback-Leibler divergence under the condition of guaranteeing coherent Bayesian updating.
\begin{theorem}\label{theorem:coherence}
Let $\pi^{\sup}_{a,C}(x, \pi_I)$ be the solution obtained by 
$$
\pi^{\sup}_{a,C}(x, \pi_I) = \arg\inf_{\pi\in\mathcal{M}} \left\{ \E_{\pi}[-L_a(\theta)] + \eta_a^{-1} D(\pi \parallel \pi_I) \right\}
$$
with data $x = \{ x_i\}_{i=1}^n$, a centring distribution $\pi_I$, and arbitrary $g$-divergence measure $D$.
Moreover, let $x$ be partitioned as $x = \{x^{(1)}, x^{(2)}\}$, for $x^{(1)} = \{x_i\}_{i\in S}$, $x^{(2)} = \{ x_j\}_{j\in \bar{S}}$, where $S,\bar{S}$ is any partition of the indices $i=1,..,n$. 
For coherence we require,
$$
\pi^{\sup}_{a,C}(x, \pi_I)  \equiv \pi^{\sup}_{a,C}\left(x^{(2)}, \pi^{\sup}_{a,C}(x^{(1)}, \pi_I)\right)
$$
That is, the solution using a partial update involving $x^{(1)}$, which is subsequently updated with $x^{(2)}$, should coincide with the solution obtained using all of the data, for any partition. Then for coherence  $D(\cdot || \cdot)$ is the Kullback-Leibler divergence.
\end{theorem}
The proof is given in Appendix \ref{app:coherence}. \qed

This theorem shows that KL is the only divergence measure to provide coherent updating of the local least favourable distribution.

\paragraph{Local sensitivity:} 
Although the framework presented here fits into \textit{global robustness} methods, it is also possible to use it to extract \textit{local robustness} measures. 
In appendix \ref{localsens} we show that the derivative at zero of least favourable expected loss w.r.t. $\lambda$ (exponential tilting parameter) is exactly the variance of the loss distribution.
This justifies the use of the variance of loss as a local sensitivity diagnostic. 

%% Introduction to admissibility and motivation for local pairwise least favourable:
\paragraph{Local Bayesian admissibility}\label{admissibility}
In a classical setting, the notion of \textit{admissibility} helps define a smaller class of actions that can then be further scrutinized in order to choose an optimal decision.  A decision is denoted inadmissible if there does not exist a $\theta$ such that its risk function (frequentist) is minimal (with respect to the other decisions) at $\theta$. We note that in a Bayesian context, because the expected loss is a single quantity used to classify actions, only the action which minimizes expected loss is admissible. However if we consider the set of posterior distributions contained within a Kullback-Leibler neighbourhood of radius $C$, then an analogous definition of admissibility can be given.

First we define the pairwise difference in expected losses of any two actions $(a, a') \in A$, the `regret' loss of having chosen $a$ instead of $a'$:
\[ L_{(a,a')}(\theta) = L_a(\theta)-L_{a'}(\theta) \]
and the corresponding least favourable pairwise distribution:
\begin{eqnarray}
\pi_{(a,a'),C}^{\sup}  & =  & \arg \sup_{\pi \in \Gamma_C} \left\{ \E_{\pi}[L_{(a,a')}(\theta)]  \right\} \nonumber \\
& = & Z_C^{-1} \pi_I(\theta) \exp\left(\lambda_{(a,a')}  [L_a(\theta)-L_{a'}(\theta)]\right) \nonumber 
% \frac{L_{a'}(\theta)}{L_a{\theta}}\right) \nonumber 
\end{eqnarray}
with expected loss $\psi_{(a, a')}^{\sup}(C) = \int_{\theta} \pi_{(a,a'),C}^{\sup}(\theta)  [L_a(\theta)-L_{a'}(\theta)] d \theta$.
\begin{definition}\label{def:Cdominated}
 We say that an action $a$ is $C^*$-dominated, or locally-inadmissible up to level $C^*$ when, 
$$
C^* := \arg\sup \{C : \psi^{\sup}_{({a, a'})}(C) < 0, ~~ \forall a' \in A \setminus {a} \}
$$
\end{definition}
If $a$ is $\infty$-dominated then it is globally inadmissible (this retrieves the classical notion of admissibility).

We note that plotting $\psi_{(a)}^{\sup}(C), a\in\mathcal{A}$ for values $C\in[0,C^*]$ does not give any information as to the admissibility of the actions $a\in\mathcal{A}$.
In order to graphically represent admissibility (inadmissibility), it is necessary to look at least favourable distributions defined over the pairwise difference in loss for two actions $a,a'$. By plotting $\psi^{\sup}_{(a,a')}$ as a function of constraint radius $C$ we can look for actions that are dominated, such that there is no $\pi \in \Gamma_C$ for which they are optimal.

\paragraph{Calibration of neighbourhoods}\label{calibration}
In most scenarios the local least favourable distribution $Z_{\lambda}^{-1} \pi_I(\theta)\exp[\lambda_a L_a(\theta)]$ will not have closed form. Moreover $\pi_I(\theta)$ will often only be 
represented as a Monte Carlo bag of samples $\{\theta_i\}_{i=1}^m \sim_{iid} \pi_I$. In this case the distribution can be approximated by using $\pi_I$ as an importance sampler (IS) leading to,
\begin{eqnarray}
\tilde{\pi}^{\sup}_{a,C} & = & \frac{1}{Z_{\lambda_a}} \sum_i w_i \delta_{\theta_i}(\theta) \nonumber \\
w_i & = & \exp[\lambda_a L_a(\theta_i)], \quad Z_{\lambda_a}= \sum_i w_i  \nonumber
\end{eqnarray}
We can then use $\tilde{\pi}^{\sup}_{a,C}$ to calculate $(\psi^{\inf}_{(a)}, \psi^{\sup}_{(a)})$. 
For a small neighbourhood size and hence small $\lambda_a$ relative to $L_a(\theta)$ this IS approximation should be reasonable. In general if $\pi_I$ is thought to be a useful model to the truth then the neighbourhood size should be kept small. However as $\lambda$ increases, the variance of the un-normalised importance weights 
will grow exponentially and the approximation error with it. In this situation the problem appears amenable to sequential Monte Carlo samplers taking $\lambda_a \ge 0$ as the ``time index'' although here we do not explore this any further. This points to the wider issue of how to choose the size of the neighbourhood $\Gamma$. 
However, in the Monte Carlo setting of this problem, the statistician can decide on candidate KL values using one or more of the following ideas:

\begin{itemize}
\item Plotting the distribution of the importance weights and deciding whether this is `degenerate'. By this we mean all the mass on one or more samples with high loss, and no mass on the lower values. This can be summarised by the variance of the weights, with the minimax solution having variance $(m-1)/m^2$.
\item Define a \textit{inequality} score for the set of importance weights. All weights $w_i=1/m$ would be perfectly equal, and the minimax solution would be perfectly unequal. We suggest considering KL values $C$ up to a $C_{max}$ defined as assigning 99\% of the mass to 1\% of the samples.
\item Use qualitative exploratory methods. For example, plot the marginal distributions of the minimax solution over dimensions of interest, i.e. which are interpretable.
\end{itemize} 
We consider that the calibration of the Kullback-Leibler divergence remains an open problem, even though this divergence is used in many applications. \cite{mcculloch1989local} proposes a general solution using a Bernoulli distribution, but it is not obvious that this translates well into a method for the calibration of any continuous distribution.

\subsubsection{Connections to other work}

The use of local least favourable distributions turns out to be connected to well known statistical techniques. We outline some examples for illustration. 

\paragraph{Predictive tempering as a local least favourable distribution} 

Consider the task, or action, to provide a predictive distribution, $\widehat{f(y | x)}$, for a future observation $y$ given covariates $x$. The local proper scoring rule in this case is known to be the self-information logarithmic loss $L(y) = - \log f(y | x)$ \cite{bernardo1994bayesian}. The conventional Bayesian solution is to report your honest marginal beliefs as $\widehat{f(y | x)} = f_I(y | x)$, where given a model parametrised by $\theta$ we have $f_I(y|x) = \int f(y | x, \theta) \pi_I(\theta) d \theta$. However this assumes that the model is true and moreover that the prediction problem is stable in time, in that the prediction probability contours do not drift (see below). Both of these assumptions may be dubious. The robust local-minimax solution above protects against misspecification and leads to 
$$
\widehat{f^*_{\sup}(y | x)} \propto f_I(y|x)^{1- \lambda},
$$
for $\lambda \in [0, 1]$. This has the form of tempering the predictive distribution, taking into account additional external levels of uncertainty outside of the modelling framework. In this way, predictive annealing can be seen as a local-minimax action.

\paragraph{Concept drift.} In data mining applications we may have access to meta-data, $t_i$, for the $i$'th observation and a belief that loss is ordered or structured by the information in $t$. For example, $t$ might index time and due to ``concept drift'' the analysis might hold greater loss in predicting using more historic collected observations, \citep[e.g. Section 3.1 in][]{hand2006classifier}, though more generally $t_i$ simply contains information relative to predictive loss. In this case the natural loss function is a weighted self-information loss, based on the empirical distribution:
$$
L(\theta) = - \sum_i \Delta(t_i) \log f(y_i; \theta).
$$
with $\Delta(t_i) \in (0,1)$ encapsulating the relative weight of log loss to the future predictive.
 
For prediction of a new observation $y^*$ given $x^*$ this leads to the robust solution as 
\begin{eqnarray*}
\widehat{f_{\sup}(y^* | x^*)}  & \propto & \int_{\theta} f(y^* | x^*, \theta) \left[ \prod_i f(y_i; \theta) \pi(\theta)] \right]  e^{- \sum_i \Delta(t_i) \log f(y_i; \theta)}  d \theta \\
 & \propto & \int_{\theta}  f(y^*|x^*, \theta) \left[ \prod_i f(y_i; \theta)^{1-\Delta(t_i)} \pi(\theta) \right] d \theta
 \end{eqnarray*}
 that can be seen to down-weight the information in $y_i$ used to predict $y^*$. For example, if $t_i$ records the time since the current prediction time then a natural penalty is $\Delta(t_i) = \exp(-\lambda t_i)$, where $\lambda$ encodes a predictive forgetting factor. For a related approach see \cite{hastie1993varying}.

\paragraph{Gibbs Posteriors and PAC-Bayesian:} Suppose you hold prior beliefs about a set of parameters $\theta$ but don't know how to specify the likelihood $f(x | \theta)$, and hence lack a model $\pi(x, \theta)$.  For example, suppose $\theta$ refers to the median of $F_X$ with unknown distribution. Suppose the task (action) is to provide your best subjective beliefs $\pi(\theta | \cdot)$ conditional on information in the data and prior knowledge. We don't have a likelihood but we could have a well defined prior hence $\pi_I(\theta) = \pi(\theta)$. In this situation there may 
be a well defined loss function on the data that we would wish to {\em{maximise}} utility  against for specifying beliefs, e.g. for the median we should take 
$$
L(\theta) = \sum_i | x_i - \theta |
$$
The distribution that {\em{minimises}} the expected loss within a certain KL divergence of the prior is given by  the local-maximin distribution, 
$$
\pi^{\inf}_{a,C} = Z_{\lambda_a}^{-1} e^{- \lambda_a \sum_i | x_i - \theta |} \pi(\theta) 
$$
This has the form of a Gibbs Posterior or an exponentially weighted PAC-Bayesian approach \citep{Zhang2006a,Zhang2006b,Bissiri2013,Dalalyan2008,Dalalyan2012}. In this way we can interpret Gibbs posteriors as local-maximin solutions in the absence of a known sampling distribution \citep{Bissiri2013}.

\paragraph{Conditional $\Gamma$-minimax priors:} 
There is a direct relationship to $\Gamma$-minimax priors when the $L_a(\theta)$ involves all the parameters in a parametric model where the posterior is
\[\pi_I(\theta) \propto \Pi_{j=1}^n f(x_j | \theta) \pi(\theta)\]
with likelihood $f(\cdot | \theta)$ and prior $\pi(\theta)$. 

Thus the posterior least favourable distribution
\[ \pi_a^{\sup}(\theta) \propto e^{\lambda_a L_a(\theta)} \Pi_{j=1}^n f(x_j | \theta) \pi(\theta) \]
can be considered a Bayes update using the minimax prior $e^{\lambda_a L_a(\theta)} \pi(\theta)$ (dropping the normalisation constant). This is an action specific prior. 

Note that the KL divergence is the only divergence to ensure this coherency, and also that the ``prior'' $ \pi^{\sup}_a(\theta)$ is data dependent if the loss function uses the empirical risk, i.e. is of the form $L_a(\theta,X)$.

\subsection{Weak Neighbourhoods}\label{dptheory}

From a Bayesian standpoint  its more natural to characterise the variation in expected loss arising over all models in some neighbourhood $\Gamma$, rather than performing minimax optimisation within the neighbourhood. In order to quantify this uncertainty and take expectations over distributions in the neighbourhood of $\pi_I$, we require a probability distribution on a set of probability measures centred on $\pi_I$. This is classically a problem in Bayesian nonparametrics, see for example \cite{hjort2010bayesian}. 
However, in a decision-theoretic context, only the functionals of the distributions $\pi\in\Gamma$ are of importance. In particular the functionals $\psi_a : \pi \rightarrow \E_{\pi}[L_a(\theta)]$ for $a\in\A$ (expected loss). It is important to note that two sequences of distributions $\pi_n,\pi_n^*$ can be infinitely divergent in Kullback-Leibler, or can remain at a finite distance in total variation metric, but weakly converge, i.e. their functionals converge, see \cite{KLnote} for an example and a further discussion of this.
Thus, if we set a nonparametric distribution $\Pi$ over measures $\pi$, that is centred at $\pi_I$: instead of studying the 'distance' between draws $\pi\sim\Pi$ and the reference distribution $\pi_I$, we can study the distance between the induced distributions $F_{a,\pi}(z)$ and $F_{a,\pi_I}(z)$, the (cumulative) distributions of loss for action $a$.
A suitable candidate distribution $\Pi$ should have wide support (to overcome the possible misspecification) and it should be possible to characterise the distance of the induced distributions $F_{a,\pi}$. The Dirichlet Process (DP) allows for exactly such a construction.

\subsubsection{Dirichlet Processes for functional neighbourhoods}

\begin{definition}
{\bf Dirichlet Process:} Given a state space ${\cal{X}}$ we say that a random measure $P$ is a Dirichlet Process on ${\cal{X}}$, $P \sim DP(\alpha, P_0)$, with concentration parameter $\alpha$ and baseline measure $P_0$ if for every finite measurable partition $\{B_1, \ldots, B_k\}$ of ${\cal{X}}$, the joint distribution of $\{P(B_1), \ldots, P(B_k)\}$ is a $k$-dimensional Dirichlet distribution $Dir_k\{\alpha P_0(B_1), \ldots, \alpha P_0(B_k)\}$.
\end{definition}
Using this definition we can then sample from distributions in the neighbourhood of  $\pi_I$ according to
$\pi \sim DP(\alpha, \pi_I)$, for some $\alpha>0$. 
In practice we can consider a draw from the DP via a constructive definition, 
\begin{equation}
\begin{aligned}
\{ \theta_i \}_{i=1}^m & \sim  \pi_I \\
 \underline{w} & \sim  {\rm{Dir}}_m(\alpha / m, \ldots,  \alpha / m),  \\
\tilde{\pi}(\theta) & :=  \sum_{i=1}^{m} w_i \delta_{\theta_i}(\theta)  \\
\end{aligned}\label{eq:constructivedef}
 \end{equation}
where the $\theta_i$'s are i.i.d. from $\pi_I$ and independent of the Dirichlet weights. As $m\rightarrow\infty$, $\tilde{\pi}$ tends to a draw $\pi\sim DP(\alpha,\pi_I)$.
This construction fits well with the Monte Carlo context, where $\pi_I$ is represented by a bag of samples $\{\theta_i\}_{i=1}^m$. If we draw multiple vectors $\underline{w}^{(1)},..,\underline{w}^{(k)} \sim {\rm{Dir}}_ m$, then in the limit $m\rightarrow\infty$, each corresponds to an independent draw from the $DP(\alpha,\pi_I)$, conditional on the atoms $\theta_i$.
In an ideal world, we would want to resample a set $\{\theta_i\}_{i=1}^m$ at each step. But this would not be feasible in practice and would defeat our purpose of constructing an ex-post methodology for analysing sensitivity. Therefore, this construction of the Dirichlet Process is more adapted than say the stick-breaking representation.

For an action $a$, the expected loss under the re-weighed draw $\tilde{\pi}$ is given by:
\begin{equation}\label{lossconstruction}
\psi_a^{\tilde{\pi}} = \sum_i w_i L_a(\theta_i)
\end{equation}
and the loss distribution by:
\[ F_{a,\tilde{\pi}}(z) = \sum_i w_i\1_{z\leq L_a(\theta_i)} (z)\]
In what follows, without loss of generality, we fix $a$ and consider the $\theta_i$ to be ordered by loss, i.e. $L_a(\theta_1)\leq...\leq L_a(\theta_m)$. Let $v_i = \sum_{j=1}^i w_i$, the cumulative summed weights, and $x_i := i/m$ for $i=1,..,m$. We also consider that the loss function $L(a,\theta)$ has undergone the following linear transformation (which does not alter the ranking of actions under expected loss):
\begin{equation}\label{losstransformation}
 L(a,\theta) \rightarrow \frac{L(a,\theta) - \min_{a,\theta} L(a,\theta)}{\max_{a,\theta} L(a,\theta)  - \min_{a,\theta} L(a,\theta)}  \end{equation}
This means each loss cdf takes values between [0,1].
We can study the $L_1$ distance between the empirical distribution\footnote{empirical in the sense that it corresponds to $\pi_I$ through i.i.d. sampling.} $F_{a,\tilde{\pi}_I}$ and the reweighed version $F_{a,\tilde{\pi}}$ which is given by:
\[\sum_{i=1}^m | v_i - x_i |\cdot [L_a(\theta_i) - L_a(\theta_{i-1})]  \]
For a fixed sample $\{\theta_i\}_{i=1}^m$, the increments $L_a(\theta_i) - L_a(\theta_{i-1})$ are also fixed, and it is possible to compute the expected difference $| v_i - x_i |$ by noting that $v_i\sim {\rm{Beta}}(x_i \alpha, (1-x_i)\alpha)$.
This is given by:
\[ \E_v\{|v_i-x_i|\} = \frac{2}{\alpha}\frac{[x_i^{x_i}(1-x_i)^{(1-x_i)}]^{\alpha}}{{\rm Beta}(x_i \alpha, (1-x_i)\alpha)}\]
As a consequence of the linear transformation given in (\ref{losstransformation}), this $L_1$ difference is bounded by 1/2.
 $\E_{\underline{w}}\{| F_{a,\tilde{\pi}} - F_{a,\tilde{\pi_I}} |\}$ is dependent on the concentration parameter $\alpha$ which controls how close the draws $F_{a,\tilde{\pi}}$ are from the reference loss distribution; increasing $\alpha$ shrinks the $L_1$ distance.
However, it is important the note that this distance will also be dependent on the form of the loss function, i.e the increments $L_a(\theta_i) - L_a(\theta_{i-1})$.

\subsubsection{Probability of optimality}

From properties of the Dirichlet Process, we know that  $\E_{\Pi}[L_a(\theta)] = \E_{\pi_I}[L_a(\theta)]$, where $\Pi$ is the nonparametric measure defined in equation (\ref{eq:constructivedef}).
Thus if an action $a$ is optimal under the criterion of posterior expected loss (taken with respect to $\pi_I$), it will remain optimal under expected loss taken with respect to $\Pi$. 
%This is because expected loss is a linear mapping. However, other non-linear mappings such as quantile loss or CVaR would not be the same under $\pi_I$ and $\Pi$. 
Instead of looking at expected loss we consider the probability that a particular action will be optimal when drawing a random $\pi\sim DP(\pi_I,\alpha)$ (and computing expected loss with respect to this random $\pi$). 
That is to say, each random draw $\pi$ will induce a distribution of loss $F_{a,\pi}$ for action $a$. The probability that $a$ is optimal will depend on the concentration parameter $\alpha$.
As the concentration parameter $\alpha\rightarrow\infty$, the random loss distribution $F_{a,\pi}$ tends to $F_{a,\pi_I}$ in probability under the $L_1$ norm, thus giving back the optimality mapping induced by $\pi_I$. This gives rise to a diagnostic graph, where the probability of optimality of each action is plotted against the parameter $\alpha$. 
The probability of optimality is non-analytical in the general case, and dependent on the form of the loss function $L(a,\theta)$. However, given a Monte Carlo representation of $\pi_I$ and thus a matrix of loss values (number of samples $\theta_i$ times number of actions) it is easy to approximate via successive draws $\underline{w} \sim Dir(\alpha/n,..,\alpha/n)$ and using the construction given in (\ref{lossconstruction}).

\subsubsection{Calibration of the Dirichlet Process extension}

Contrary to the methodology proposed in section \ref{KLanalytic}, using a Dirichlet Process as a nonparametric model extension to test robustness gives a framework which is not action specific. However, it also relies on a free parameter $\alpha$ which needs to be calibrated in a principled manner. In order to do this, we define a reference action $a^*$, such that the loss distribution is given by $F_{a^*,\pi}(z) = z$, for $z\in[0,1]$\footnote{After the linear transformation given in (\ref{losstransformation}).}.

With the construction given in \eqref{lossconstruction}, if we draw weights $ \underline{w} \sim {\rm{Dir}}_m(\alpha / m, \ldots,  \alpha / m)$ and use uniform loss intervals, i.e. reweighing the distribution $F_{a^*,\pi}(z)$, it is then possible to compute $95\%$ confidence intervals for draws from a Dirichlet Process for a specific value of the parameter $\alpha$. Figure \ref{lossconfplots} shows an approximation of these confidence intervals for a series of values $\alpha$ whose $\log_{10}$ values are integers from 0 to 4. It is clear that low values of $\alpha$ (between 0 and 10 for example) imply a very low trust in the model, as the draws can vary hugely from the reference distribution. However the statistician might want the decisions to be robust to higher values of $\alpha$ where the draws are much tighter.
In section \ref{breastcancer} we illustrate the use of this method. 
\begin{figure}
\centering
\includegraphics[scale=.45]{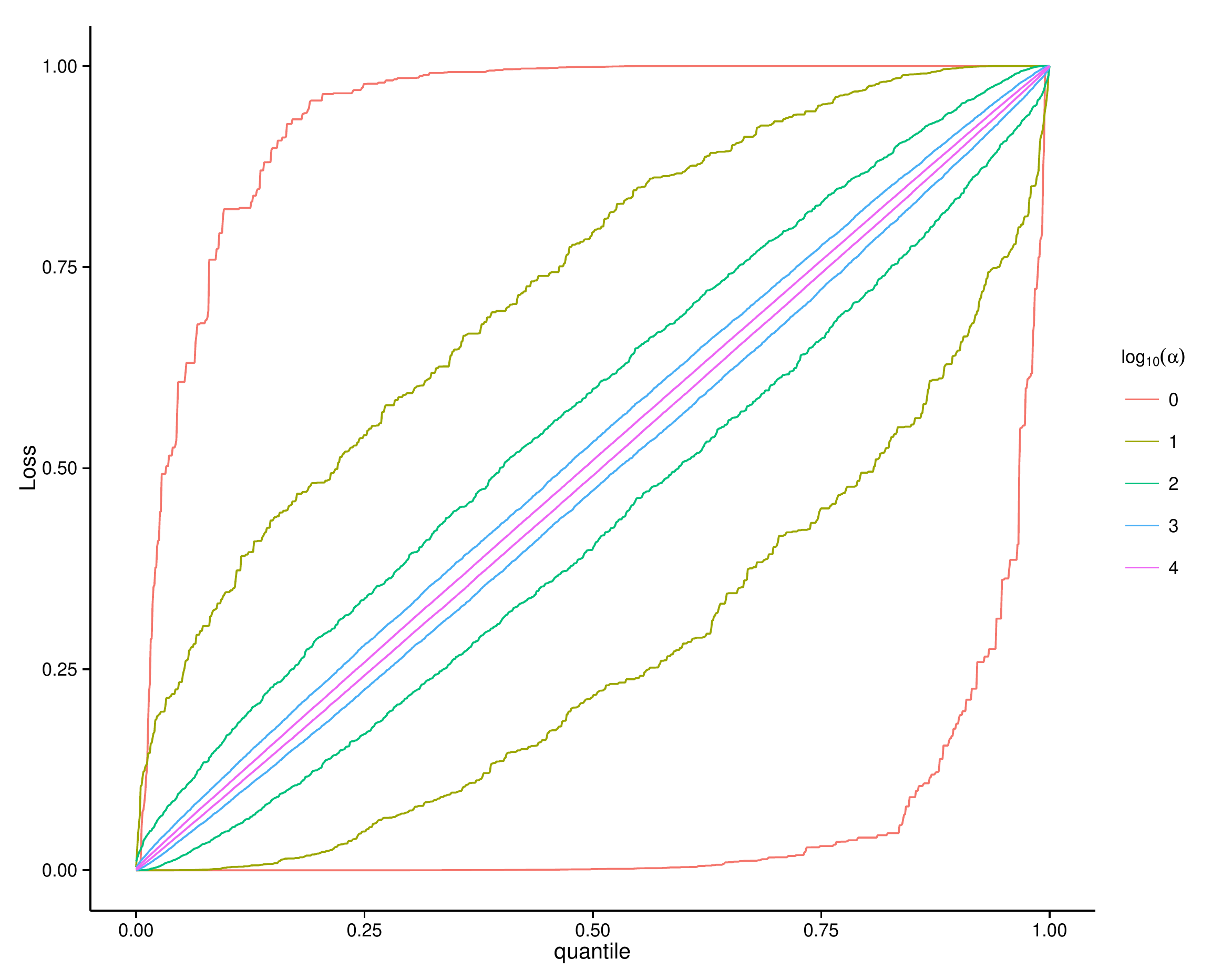}
\caption{95\% confidence intervals around the reference action $a^*$ for different values of the concentration parameter $\alpha$. From widest to tightest: $\log_{10}(\alpha)$ values are 0,1,2,3,4.}
\label{lossconfplots}
\end{figure}

\section{Applications}\label{applications}

\subsection{Synthetic Example, continued}

\begin{figure}
\centering
\includegraphics[scale=.45]{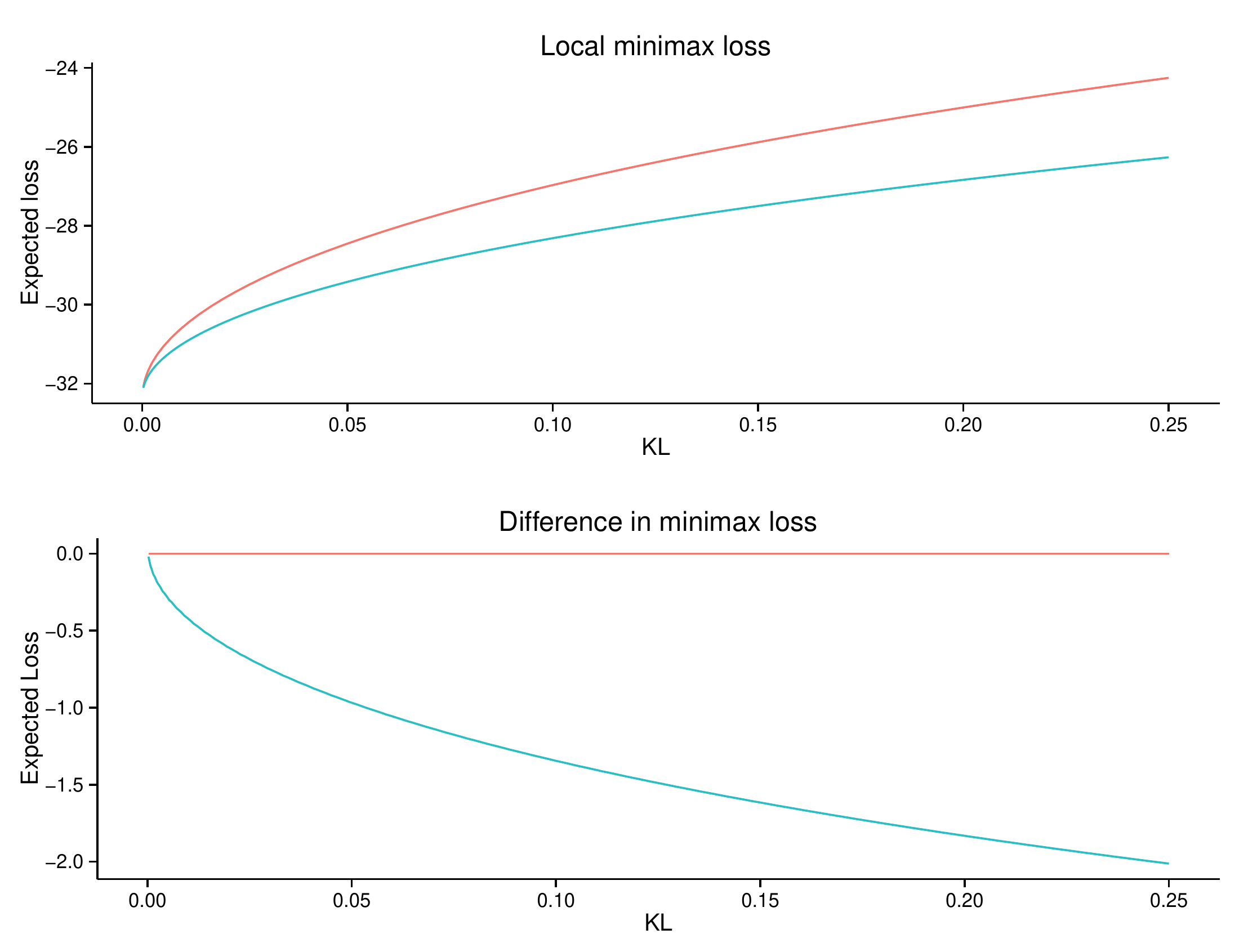}
\caption{Diagnostic plots for the local least favourable distribution, comparing vaccination (blue) to status quo (red). Top: local least favourable expected loss; Bottom: difference in minimax between non optimal and optimal actions.}
\label{SynKLplots}
\end{figure}
To illustrate the ideas from sections \ref{KLanalytic} and \ref{dptheory}, we continue with the vaccination vs. status quo decision problem given in section \ref{syntheticapp}. Figure \ref{SynKLplots} plots the expected loss of each action under the local least favourable distribution, as a function of the size $C$ of the KL ball\footnote{Using the function \textit{preliminaryAnalysis} from our package \textit{decisionSensitivityR}}. We observe that for very small value of KL the status quo action is no longer optimal. This is because of its much higher variance of loss. However, figure \ref{Synprobplots} shows the probability of optimality for each action as a function of the concentration parameter $\alpha$ (log scale). This highlights that in fact, the system is decision robust (see section \ref{discussion} for a discussion on decision robustness vs. loss robustness)
\begin{figure}
\centering
\includegraphics[scale=.35]{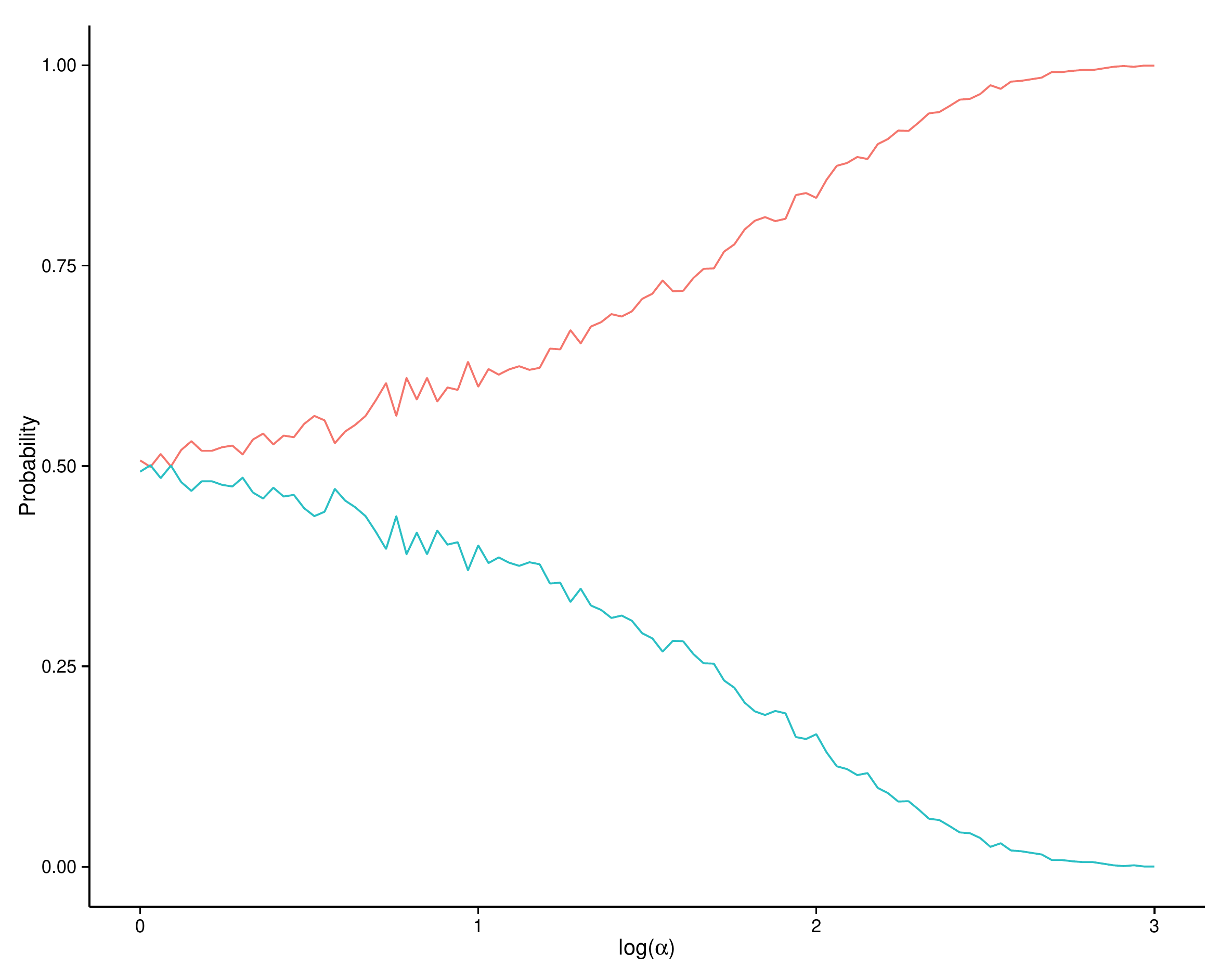}
\caption{Probability of optimality of the two actions (red: status quo, blue: vaccination) under a nonparametric extended model using a $DP(\pi_I,\alpha)$. The $\alpha$-values are plotted on a log scale.}
\label{Synprobplots}
\end{figure}
The vaccination vs. status quo decision problem is synthetic but it allows us the illustrate the diagnostic plots based on the formal methodology presented in section \ref{formalmethods}. We see that the two diagnostic methods respectively based on KL balls and on Dirichlet process extensions highlight different ways in which the optimal action can be sensitive to model misspecification. The local least favourable distribution concentrates on the high loss values of each action, thus making the vaccination action preferable for very small KL values (see figure \ref{SynKLplots}, top right). However using a Dirichlet model extension, the decision system is robust to symmetrical perturbations around $\pi_I$. This is shown in figure \ref{Synprobplots} where even for small values of $\alpha$ (shown on a log scale in the plot), the vaccination action (blue) is not more probably optimal than the status quo (red). This shows strong stability to these symmetrical perturbations.
Taking a decision as to whether to trust the model or not would be context dependent.

\subsection{Optimal Screening Design for Breast Cancer}\label{breastcancer}

Public health policy is an area where the application of statistical modelling can be used to optimally allocate resources. Breast cancer screening for healthy women over a certain age to detect asymptomatic tumours, is a hotly debated and controversial issue for which it is difficult to fully quantify the benefits. A recent independent review (\cite{independent2012benefits}), commissioned by Cancer Research UK and the Department of Health (England) concludes that only a randomised clinical trial would fully resolve this issue. This is of course the gold standard which permits causal inference. However a primary issue is determining the optimal screening schedule, consisting of a starting time $t_0$ (age of first screen), and a frequency $\delta$ for subsequent screens. It is of course sharply infeasible to trial all combinations of schedules $(t_0,\delta)$. An optimal trial design however can be constructed using a statistical model of disease progression throughout a population. \cite{Parmigiani} proposed using a semi-Markov process consisting of four states which generalises to any chronic disease characterised by an asymptomatic stage. All individuals start in state $A$, disease-free. They then transition either to the absorbing state $D$ (death) or contract the disease, modelled by a transition to state $B$, the pre-clinical stage. This is followed by a transition to either the clinical stage of the disease or death. It was assumed that each transition happens after a time $t$ with the following densities:
\begin{equation}
  \begin{aligned}
t_D \sim  h(t|\alpha,\beta) &= \mathrm{Weibull}(\alpha,\beta)  \\
t_B \sim f(t|\mu,\sigma^2) &= \mathrm{LogNormal}(\mu,\sigma^2) \\
t_C \sim q(t|\kappa,\rho) &= \mathrm{LogLogistic}(\kappa,\rho) 
\label{piImodel}
  \end{aligned}
\end{equation}

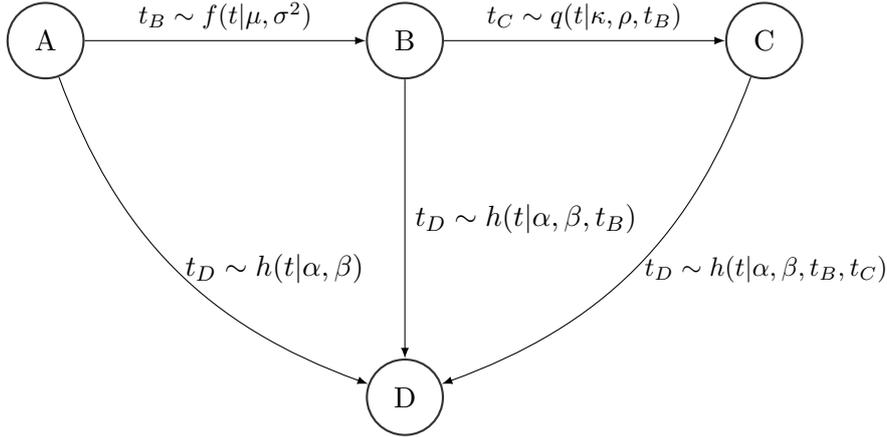
\begin{figure}
\centering
\begin{tikzpicture}
\tikzstyle{main}=[circle, minimum size = 10mm, thick, draw =black!80, node distance = 37mm]
\tikzstyle{connect}=[-latex, thin]
\tikzstyle{box}=[rectangle, draw=black!100]
  	\node[main, fill = white!100] (A) {A};
 	 \node[main] (B) [right=of A] {B};
 	 \node[main] (D) [below=of B] {D};
 	 \node[main] (C) [right=of B] {C};
 	 \path[every node/.style={font=\sffamily\small}]
	 	(A) edge [connect] node [above] {$t_B\sim f(t|\mu,\sigma^2)$} (B) 
        		(B) edge [connect] node [above] {$t_C \sim q(t|\kappa,\rho,t_B)$} (C) 
        		(C) edge [connect, bend left=25] node [right] {$t_D \sim h(t|\alpha,\beta,t_B,t_C)$} (D);
 	\path
        		(A) edge [connect, bend right=25] node [right] {$t_D \sim h(t|\alpha,\beta)$} (D);
	\path
		(B) edge [connect] node [right] {$t_D \sim h(t|\alpha,\beta,t_B)$} (D) ;
\end{tikzpicture}
\caption{Graphical model of the transitions and transition densities between states.}
\label{GraphicalModel}
\end{figure}

Figure \ref{GraphicalModel} shows a graphical model of the four state semi-Markov process with transition densities. An individual is characterised by the triple $t=(t_B, t_C, t_D)$ where the symptomatic stage of the disease is contracted only when $t_D < t_B+t_C$ (assuming that all individuals will contract the disease if they lived long enough). For a screening schedule $a=(t_0,\delta)$ the loss function is defined as follows (a function of the times $t=(t_B,t_C,t_D)$):

\begin{equation}\label{bclossfunction}
L(a,t) = r\cdot n_a(t)  + \mathbbm{1}_{C} 
\end{equation}

\begin{figure}[ht]
\centering
\includegraphics[scale=.4]{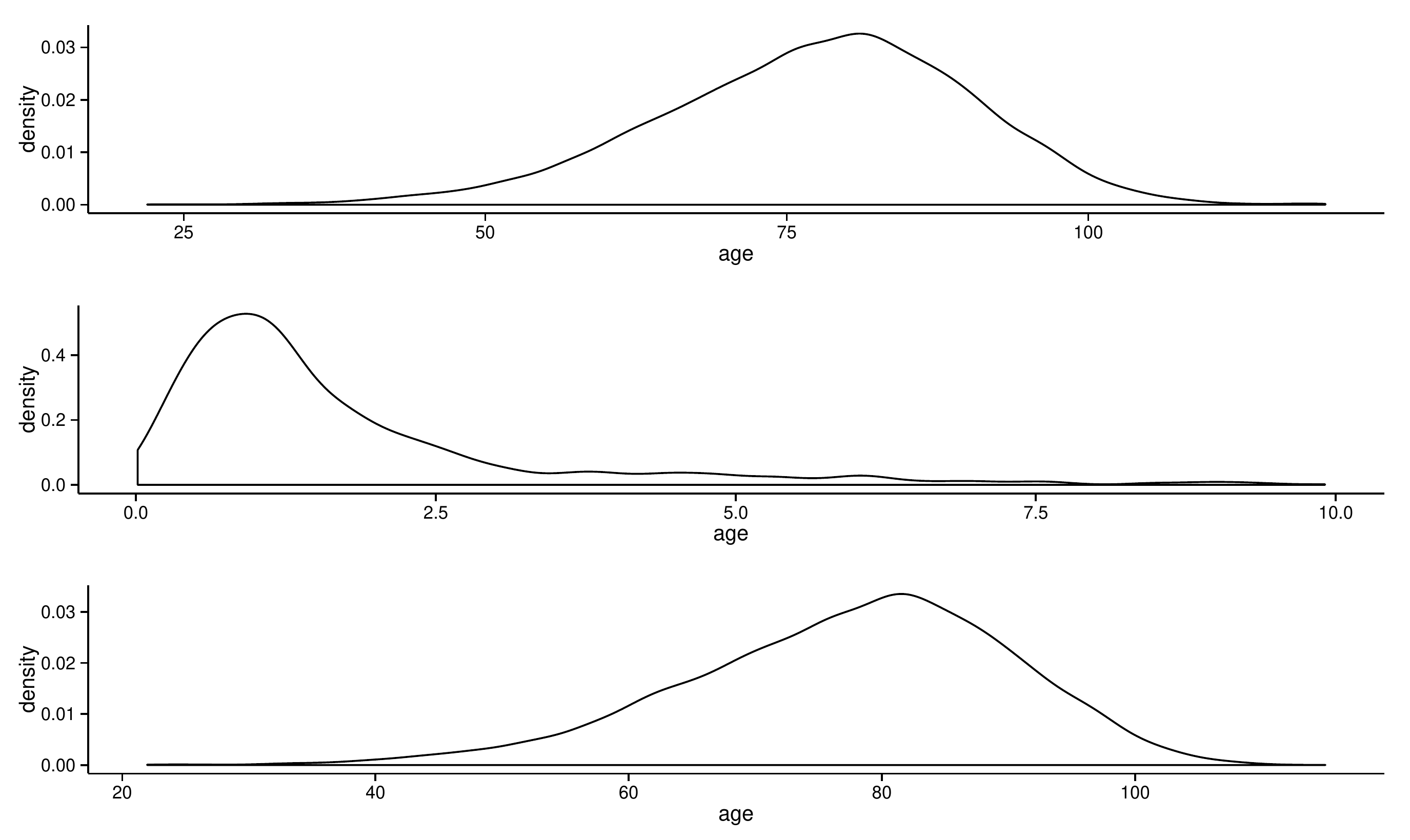}
\caption{Probabilistic model of transition times: (from top to bottom) marginal densities of transition times to preclinical stage, transition from preclinical to clinical stage, and death times.}
\label{BCsampletimes}
\end{figure}
\begin{figure}[ht]
\centering
\includegraphics[scale=.5]{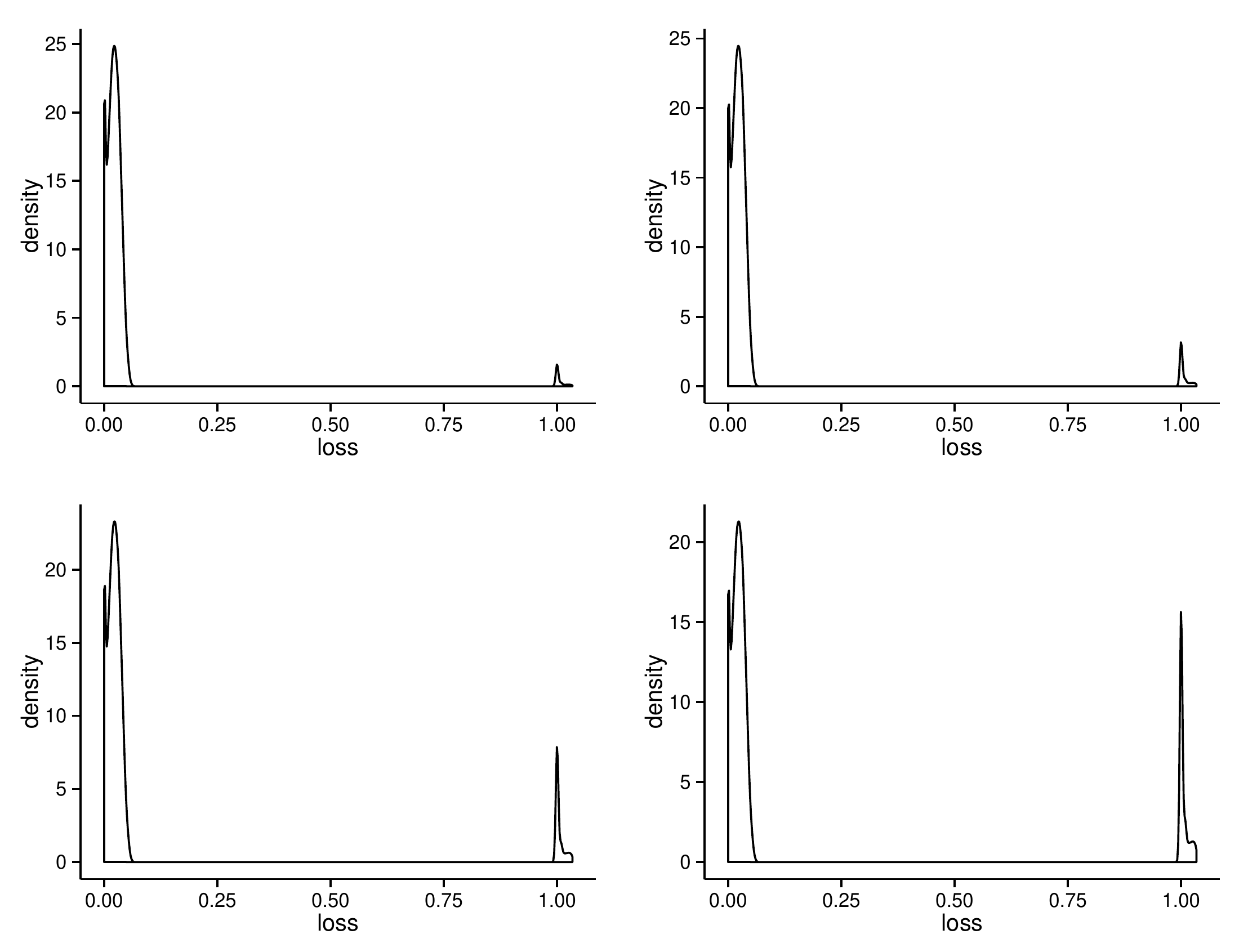}
\caption{Top left: loss density for the optimal action (start at 59, frequency 15 months) under the approximating model $\pi_I$ given in (\ref{piImodel}). Going from top right to bottom right: loss densities, for the same action, under the local least favourable distribution at KL divergences equivalent to a reassignment of mass of 2, 5 and 10\%, respectively. These are approximate: 0.008, 0.08, 0.3. The effect can be seen as increasing the mass put onto high loss events.}
\label{BClossplots}
\end{figure}
where $n_a$ is the number of screening schedules an individual will receive during their lifetime, until they die or enter into the asymptomatic stage of the disease. $\mathbbm{1}_C$ is the indicator function, taking value 1 for the event that the pre-clinical tumour is not detected by screening or occurs before $t_0$, and zero otherwise. $r$ trades off the cost of one screen against the cost incurred by the onset of the clinical disease. Each screen has an age-dependent false-negative rate modelled with a logistic function:
\[ \beta(t) = \frac{1}{1+ e^{-b_o-b_1(t-\tilde{t})}}\]
where $\tilde{t}$ is the average age at entry in the study group.
To simulate transition times for individuals from this model, we used 2000 posterior parameter samples for $\theta=(\mu,\sigma^2,\kappa,\rho,b_0,b_1)$ given in the supplementary materials of \cite{BreastCancer}. This is based on data from the HIP study \cite{HIPstudy}. Figure \ref{BCsampletimes} shows the estimated marginal densities for $10^4$ sampled times for each transition event\footnote{We calibrate the Weibull distribution with values $\alpha=7.233,\beta=82.651$ which are the values used in \cite{Parmigiani}}. 
To carry out an ex-post analysis of this model, we first considered 32 alternative schedules, consisting of all combinations of starting ages taken from the set $\{55,57,59,61,63,65,67,69\}$ (years) and screening frequencies of $\{9,12,15,18,24\}$ (months). This choice of screening schedules is mainly illustrative for our purposes: the optimal schedule will heavily depend on the choice of $r$ (trade-off ratio in equation \ref{bclossfunction}) which we do not attempt to justify (the value $10^{-3}$ was taken from the section 4.5 of \cite{ruggeri2005robust} where the authors also considered this application). In order for the plots to be legible, we selected the top 6 schedules \footnote{Given in order of increasing expected loss these are: (59,15), (55,15), (57,15), (61,15), (55,18) and (59,18).}(as ordered by expected loss under the reference model) for analysis. However, there is no reason not to analyse a greater number of schedules other than for clarity in plotting.
The top left plot in figure \ref{BClossplots} shows the loss density plot of the optimal action corresponding to the schedule $a=(t_0=59,\delta =15)$ (units in years and months) and a trade-off parameter $r=10^{-3}$. The other three plots show the corresponding loss density for the minimax distributions at KL values equivalent to reassigning 2,5 and 10\% of the mass, respectively. The effect can be see as transferring the mass from left to right, i.e. from low loss to high loss. 
The losses incurred for a particular schedule $a=(t_0,\delta)$ can be seen to be highly bimodal. Most of the population do not contract the disease and therefore contribute a loss of $r\cdot n_a$ (cost of screen times number of total screens during lifetime ). The loss contributed by those who do contract the clinical stage of the disease is of magnitude $1/r$ greater by definition. 
%%%%%%%%%%%%%%%%%%%%%%%%%%%%%%%%%%%%%%%%%%%%%%%%%%

%%%%%%%%%%%%%%%%%%%%%%%%%%%%%%%%%%%%%%%%%%%%%%%%%%
\begin{figure}[ht]
\centering
\includegraphics[scale=.45]{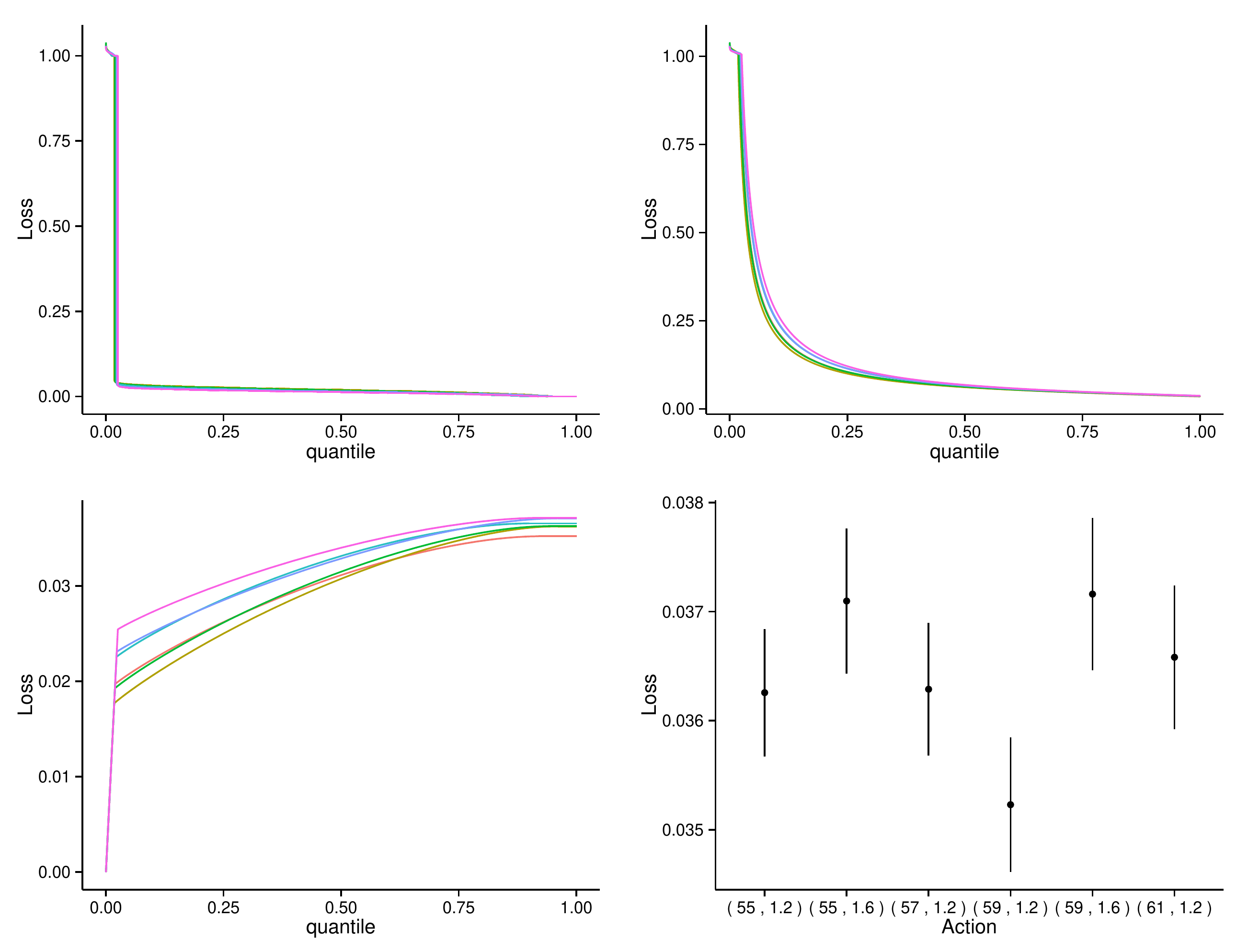}
\caption{Model diagnostic plots. From top left to bottom right: inverse loss distributions of the 6 actions (all very close in shape); Upper trimmed mean loss which differentiates the actions by showing the higher downside in some schedules; conditional expected loss; estimates of expected loss centred inside intervals of two standard deviations. We see that the expected loss $\psi_{a,\pi_I}$ for all actions is driven by low probability, high loss events (shape of CEL plot).}
\label{BClossdiagnostics}
\end{figure}
%%%%%%%%%%%%%%%%%%%%%%%%%%%%%%%%%%%%%%%%%%%%%%%%%%

Figure \ref{BClossdiagnostics} gives four diagnostic plots for the loss distribution: inverse loss distribution, the Value at Risk, the Conditional Value at Risk and the Conditional Expected Loss. These are defined in section \ref{diagnostics} and are shown here with the schedules (decisions) aforementioned. The Conditional Expected Loss plot very clearly shows that the expected loss values are driven by low probability events (around 10\% of the mass).

\begin{figure}[ht]
\centering
\includegraphics[scale=.5]{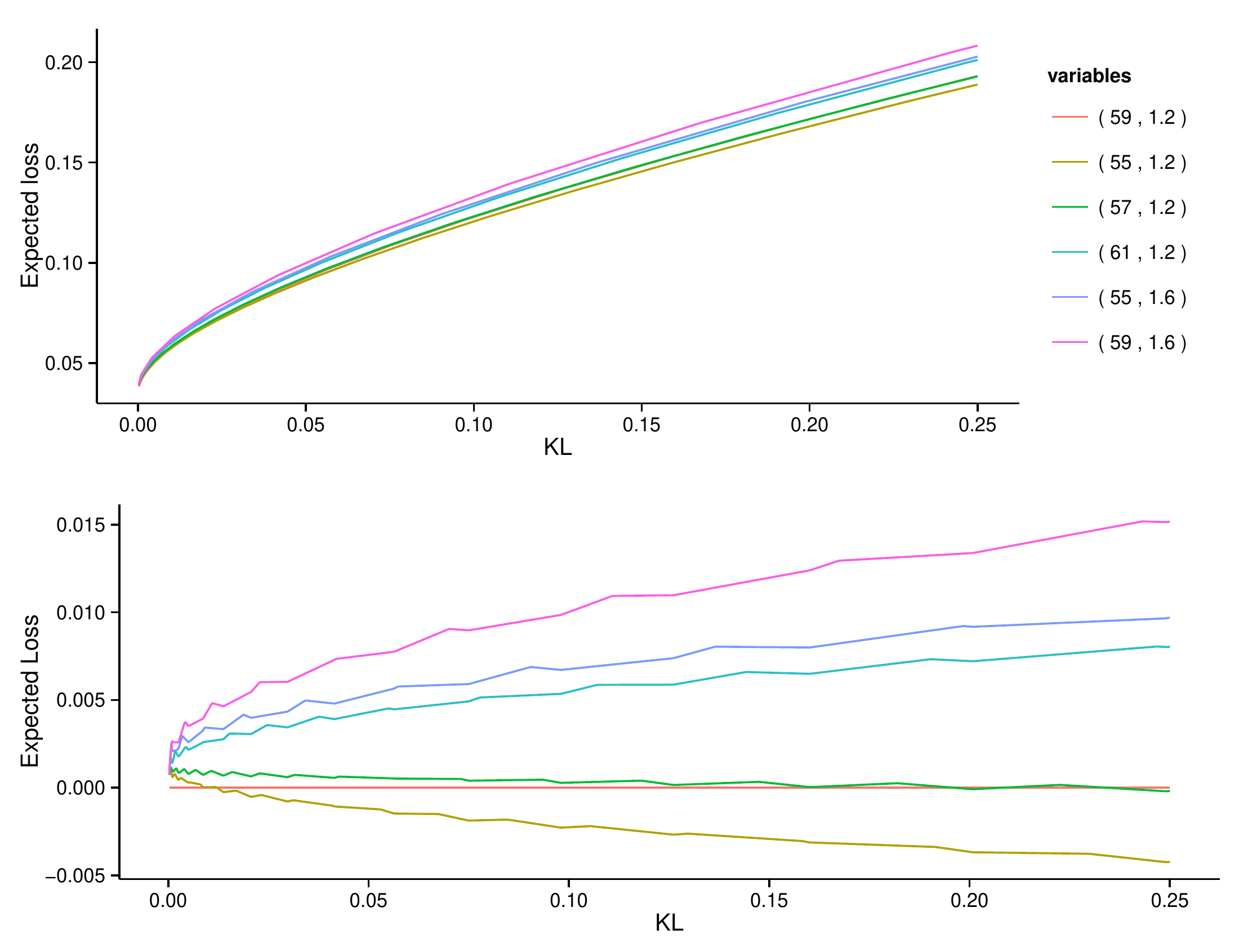}
\caption{Diagnostic plots for local least favourable distribution. Top: plot of minimax expected loss versus the size $C$ of the KL neighbourhood; Bottom: difference between the minimax expected loss of each action and that of the optimal action $a^*$.}
\label{BCKLdiagnostics}
\end{figure}

The diagnostic plot in figure \ref{BCKLdiagnostics} which based on the theory given in section \ref{KLanalytic} confirms that the decision system is sensitive to small changes in the model. The difference in expected loss under the local least favourable distributions between action (55,15) and the optimal action is negative for small KL values (bottom plot, figure \ref{BCKLdiagnostics}). Hence small perturbations (in KL divergence) changes the optimality of the actions. This is also apparent from figure \ref{localadmissibility}, where we plot the local admissibility of the optimal action under $\pi_I$ (see section \ref{admissibility}. For very small neighbourhoods of in KL, the optimal Bayes action is no longer locally admissible. 

\begin{figure}[ht]
\centering
\includegraphics[scale=.5]{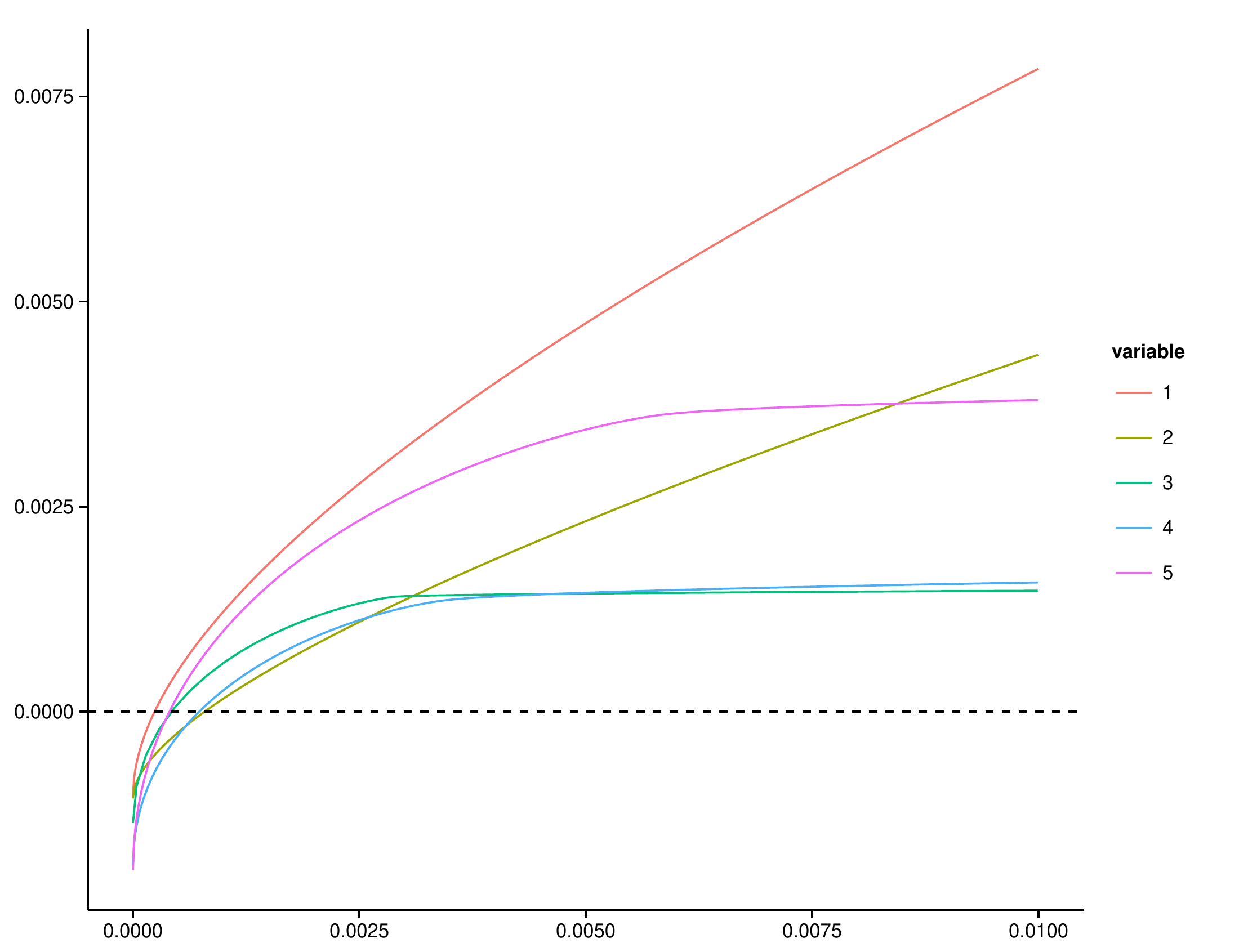}
\caption{Local Bayesian admissibility plot.}
\label{localadmissibility}
\end{figure}
%
% DP Plots
\begin{figure}
	\centering
 	\includegraphics[scale=.4]{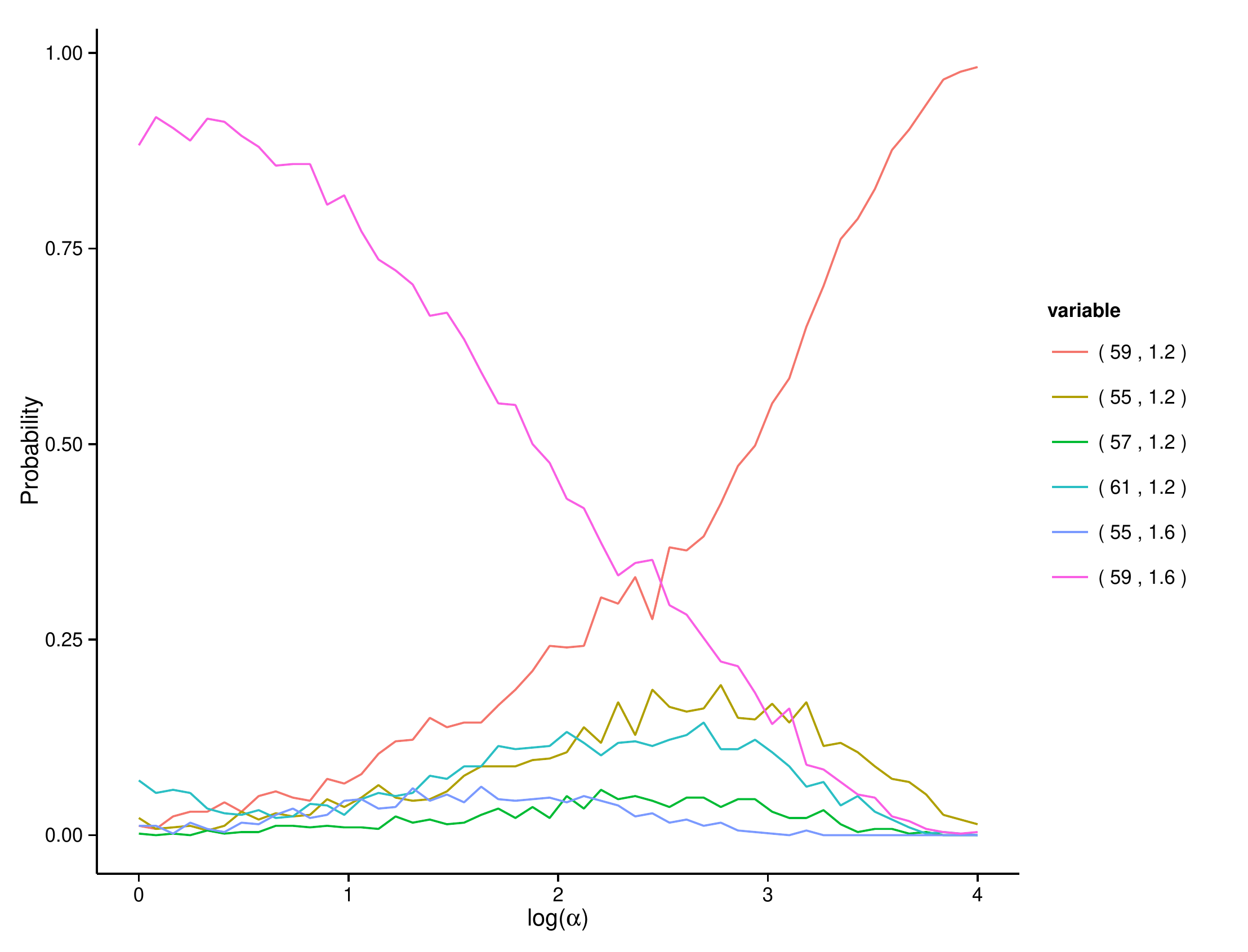}
\caption{Probability of being optimal for the top six actions selected by expected loss. The $\alpha$-values are plotted on a $log_{10}$ scale. The legend gives the actions ordered by increasing expected loss. We observe the optimal action (under $\pi_I$) becomes most likely for $\alpha\ge 10^3$. Below that value, the most probable is (59,18), i.e. the least optimal under $\pi_I$.}
\label{BCproboptimality}
\end{figure}
As a final diagnostic plot we look at the probability of optimality under the Dirichlet extension model (section \ref{dptheory}).
Figure \ref{BCproboptimality} shows the probability for varying values of concentration parameter $\alpha$ of each action being optimal. We see that for large values of $\alpha$ (greater than $10^4$) we recover the optimal action under $\pi_I$. However, for smaller values, the least optimal under $\pi_I$ (of the 6 selected) has a higher probability of being optimal. This shows the lack of decision robustness in this problem, mainly due to flatness of the loss surface.

This application highlights an interesting distinction that must be made when considering model misspecification in a decision-theoretic context. The loss surface is very flat for changes in screening schedule. That is to say, there is little relative difference in expected loss between similar screening schedules. This is also noted by \cite{ruggeri2005robust} in their analysis of the problem. This particular application is robust to changes in the model (in an expected loss sense) but not however decision robust. I.e. small perturbations to the model will change the optimality of an action $a^*$. We discuss this idea further in section \ref{discussion}.

\section{Conclusions}\label{discussion}

The goal of this article is to aid decision makers by providing statistical methods for exploring sensitivity to model misspecification. We hope this will generate further debate and research in this field. The increase in complex high-dimensional data analysis problems, ``big-data'', has driven a corresponding rise in approximate probabilistic techniques. This merits a reappraisal of existing diagnostics and formal methods for characterising the stability of inference to known approximations.

In one sense, the formal methods presented in Section 4 can be considered as extensions of the conditional $\Gamma$-minimax approach \citep{vidakovic2000gamma}. However, we advocate that the neighbourhood should be defined with respect to the marginal distribution on only those elements that enter into the loss function. Moreover we have shown that the Kullback-Leibler divergence is the only coherent measure to use for $\Gamma$ in local-minimax updating. Further motivation for using KL is given in Chapters 1 and 9 in \cite{hansen2008robustness}. To this we might add, that KL$(p \parallel q)$ is invariant to re-parametrisation; that in information theory it represents the number of bits of information needed to recover $p$ from model $q$; it represents the expected log-loss in using $q$ to approximate $p$ when using proper local scoring rules \citep{bernardo1994bayesian}; and KL bounds the L1 divergence KL$(p \parallel q) \le \parallel p - q \parallel_1$. However, none of this provides a constructive approach for choosing the KL radius $C$. In chapter 9 of \cite{hansen2008robustness}, the authors suggest using detection error probabilities to calibrate the size of the neighbourhood $\Gamma$. This stems from the concept of statistically indistinguishable models given a finite data sample of size $N$. Using model selection principles based on likelihood ratio tests, the user determines a plausible probability (a function of the radius $C$) of selecting the wrong model given the available data, and then inverts this value to find $C$ (by simulation). Although this is a principled method, in many cases even the detection error probability could be difficult to calibrate. We propose using the distribution of the tilting weights to compute an inequality score or the variance to find plausible values of $C$.   

In terms of implementation, we showed that the formal approaches have simple numerical solutions via re-weighted Monte Carlo samples drawn from the approximating model; using exponentially tilted weights for the local-minimax solution and stochastic Dirichlet weights for the marginal loss distribution.

To complete this discussion of robustness under model misspecification, it is important to note the distinction between ``decision robustness'' and ``loss robustness'' as discussed in \cite{Kadanediscussionberger}. A system is said to be decision robust if perturbations to the model do not effect the optimality of an action $\hat{a}$. On the other hand, it is said to be loss robust, if those perturbations do not effect the overall expected loss of the action $\hat{a}$ (in a relative sense). It is clear that a decision system can have one property without the other. Which is more desirable will be highly context dependent. Throughout the article we have taken the loss function to be known. However, it is clear that loss misspecification is also an important element of robust decision making. Further work is needed to develop a unified approach for dealing with this. Our framework ignores misspecification in the loss function. Certain loss functions are often chosen for computational ease or because they posses other desirable properties such as convexity. Also, elicitation of the true loss function can be difficult (for an example see the application discussed in section \ref{breastcancer}). Hence for completeness, a robustness analysis of a decision system must take this into account. \cite{ruggeri2005robust} (pages 636-639) provides some further discussion and references.

\section*{Acknowledgements}
%We are grateful to David Draper and Dimitris Fouskakis for providing us with the data from the RAND study on quality of hospital care.  
We thank Luis Nieto-Barajas, George Nicholson and Tristan Gray-Davies for many helpful discussions and suggestions.
Watson was supported by an EPRSC grant from the Industrial Doctorate Centre (SABS-IDC) at Oxford University and by Hoffman-La Roche. 
Holmes gratefully acknowledges support for this research from the Oxford-Man Institute, the EPSRC, the Medical Research Council and the Wellcome Trust.

\appendix
\newpage

\section{Proof of Theorem \ref{theorem:coherence} in Section \ref{formalmethods}.} {\em{Reproduced and amended from}} \citep{Bissiri2013}.\label{app:coherence}

Assume that $\Theta$ contains at least two distinct points, say $\theta_1$ and $\theta_2$. Otherwise, $\pi$ is degenerate and the thesis is trivially satisfied.
To prove this theorem, it is sufficient to consider the case $n=2$ and a very specific choice for $\pi$, taking
$\pi=p_0\delta_{\theta_1}+ (1-p_0)\delta_{\theta_2},$
where $0< p_0< 1$.
Any probability measure absolutely continuous with respect to $\pi$ has to be equal to
$p\delta_{\theta_1}+(1-p)\delta_{\theta_2}$, for some $0\leq p\leq 1$.
Therefore, in this specific situation,
the cost function, $l(\cdot) = \left\{ \E_{\pi}[-L(\theta)] + \lambda^{-1} g(\pi \parallel \pi_I) \right\}$, to be minimised becomes:
\begin{equation*}
\begin{split}
l(p,p_0,L_I)&:=p\,L_I(\theta_1)\, +\,(1-p)\,L_I(\theta_2)\\
&\phantom{=}+\, p\,g\left(\frac{p}{p_0}\right)\,+\,(1-p)\,g\left(\frac{1-p}{1-p_0}\right),
\end{split}
\end{equation*}
where $g$ is a divergence measure, $L_I(\theta_i)=L(\theta_i,I_1)+L(\theta_i,I_2)$ for data $I=(I_1,I_2)$ and $L_I(\theta_i)=L_1(\theta_i,I_j)$ for $I=I_j$, $i,j=1,2$.
Denote by $p_1$ the probability $\pi_{I_1}(\{\theta_1\})$, i.e. the minimum point of $l(p,p_1,L_{(I_1,I_2)})$
as a function of $p$, and by $p_2$ the probability $\pi_{(I_1,I_2))}(\{\theta_1\})$.
By hypotheses, $p_2$ is the unique minimum point of both loss functions
$l(p,p_1,L_{I_2})$ and $l(p,p_0,L_{(I_1,I_2)})$.
Again by hypothesis, we shall consider only those functions
$L_{I_1}$ and $L_{I_2}$
 such that each one of the functions
$l(p,p_0,L_{I_1})$, $l(p,p_1,L_{I_2})$, and $l(p,p_0,L_{(I_1,I_2)})$,
as a function of $p$, has a unique minimum point, which is $p_1$ for the first one and $p_2$
for the second and third one.
The values $p_1$ and $p_2$ have to be strictly bigger than zero and strictly smaller than one: this was proved by
Bissiri and Walker (2012) in their Lemma 2.
%\citet{BissiriWalker09} in their Lemma 2.
Hence, $p_1$ has to be a stationary point of
$l(p,p_0,h_{I_1})$
and $p_2$ of both the functions
 $l(p,p_1,L_{I_2})$ and $l(p,p_0,L_{(I_1,I_2)})$.
Therefore,
\begin{align}\label{f: loss.1}
g'\left(\frac{p_1}{p_0}\right)\,-\,g'\left(\frac{1-p_1}{1-p_0}\right)\,&=\,
L_{I_1}(\theta_2)\, -\, L_{I_1}(\theta_1),\\
\label{f: loss.2}
\, g'\left(\frac{p_2}{p_0}\right)\,-\,g'\left(\frac{1-p_2}{1-p_0}\right)\,&=\,
L_{(I_1,I_2)}(\theta_2)\, -\, L_{(I_1,I_2)}(\theta_1),\\
\label{f: loss.3}
\, g'\left(\frac{p_2}{p_1}\right)\,-\,g'\left(\frac{1-p_2}{1-p_1}\right)\,&=\,
L_{I_2}(\theta_2)\, -\, L_{I_2}(\theta_1).
\end{align}
Recall that
$L_{(I_1,I_2)}=L_{I_2}+L_{I_1}$.
Therefore,
summing up term by term \eqref{f: loss.1} and \eqref{f: loss.3},
and considering \eqref{f: loss.2}, one obtains:
\begin{equation}\label{f: equation}\begin{split}
g'&\left(\frac{p_2}{p_0}\right)\,-\,g'\left(\frac{1-p_2}{1-p_0}\right)\\
&\phantom{XXX}=\,g'\left(\frac{p_1}{p_0}\right)\,-\,g'\left(\frac{1-p_1}{1-p_0}\right)\, +\,
g'\left(\frac{p_2}{p_1}\right)\,-\,g'\left(\frac{1-p_2}{1-p_1}\right).
\end{split}\end{equation}

Recall that by hypothesis \eqref{f: loss.1}--\eqref{f: loss.3} need to hold for every two functions
$L_{I_1}$ and $L_{I_2}$
arbitrarily chosen with the only requirement that $p_1$ and $p_2$ uniquely exist.
Hence, \eqref{f: equation} needs to hold for every $(p_0,p_1,p_2)$ in $(0,1)^3$.
By substituting $t=p_0$, $x=p_1/p_0$ and $y=p_2/p_1$, \eqref{f: equation} becomes
\begin{equation}\label{f: equation.}\begin{split}
g'&\left(xy\right)\,-\,g'\left(\frac{1-txy}{1-t}\right)\\
&\phantom{XXX}=\,g'(x)\,-\,g'\left(\frac{1-tx}{1-t}\right)\, +\,
g'\left(y\right)\,-\,g'\left(\frac{1-txy}{1-tx}\right),
\end{split}\end{equation}
which holds for every $0<t<1$, and every $x,y>0$ such that $x<1/t$ and $y<1/(xt)$.
Being $g$ convex and differentiable, its derivative $g'$ is continuous. Therefore,
letting $t$ go to zero, \eqref{f: equation.}  implies that
\begin{equation}\label{f: equation+}
g'\left(xy\right)
=\,g'(x)\, +\,
g'\left(y\right)\,-\,g'(1)
\end{equation}
holds true for every $x,y>0$. Define the function
$\varphi(\cdot)=g'(\cdot) -g'(1)$.
This function is continuous, being $g'$ such,
and by \eqref{f: equation+},
$\varphi(xy)=\varphi(x)\, +\,\varphi(y)$ holds for every $x,y>0$. Hence, $\varphi(\cdot)$ is $k\ln(\cdot)$ for some $k$, and therefore
\begin{equation}\label{f: log}
 g'(x)\
=\ k\,\ln (x)\ +\ g'(1),\end{equation}
where
\(k\,=\,(g'(2)\,   -\, g'(1))/\ln(2)\).
Being $g$ convex, $g'$ is not decreasing and therefore
$k \geq 0$. If $k=0$, then $g'$ is constant, which is impossible, otherwise, for any $h_I$,
$p_1$ satisfying \eqref{f: loss.1} either would not exist or would not be unique. Therefore, $k$ must be positive.
Being $g(1)=0$ by assumption, \eqref{f: log} implies that
$g(x)\:=\:k\, x\ln(x)\, +\, (g'(1)-k) (x-1)$.
Hence,
\[   g(\pi_1,\pi_2)=  k\int \ln \bigg(\frac{\diff \pi_1}{\diff \pi_2}\bigg) \diff \pi_1 \]

holds true for some $k>0$ and
for every couple of measures $(\pi_1, \pi_2)$ on $\Theta$ such that $\pi_1$ is absolutely continuous with respect to $\pi_2$.

\section{Glossary of Terms}\label{app:glossary}

\begin{table*}[htb]

\begin{tabular}{c c p{13cm}}
Notation && Definition\\
\toprule
$\Theta$ && Parameter space describing the uncertainty in the 'small world' of interest.\\[3pt]
$a\in\mathcal{A}$ && Set of actions or alternatives.\\[3pt]
$L(a,\theta)$ or $L_a(\theta)$ && Loss function defined as mapping $\mathcal{A}\times\Theta\rightarrow \mathbb{R}^+$\\[3pt]
$L_{(a,a')}(\theta)$ && Regret loss function: $L_a(\theta)-L_{a'}(\theta)$\\[3pt]
$\pi_I$ && The approximating or reference model. This could be a Bayesian posterior, or just any distribution over the uncertainty $\Theta$.\\[3pt]
$C$ && The radius of the Kullback-Leibler ball centred at $\pi_I$\\[3pt]
$\lambda_a(C)$ && Exponential tilting parameter given in equation \ref{eq:pi_sup} for action $a$ corresponding to least favourable distribution in KL ball of radius $C$\\[3pt]
$\Gamma_C$ && Set of distributions $\pi$ satisfying KL$(\pi || \pi_I) \leq C$ (KL ball) \\[3pt]
$\Gamma_C^{\text{rev}}$ && Set of distributions $\pi$ satisfying KL$(\pi_I || \pi) \leq C$ (reverse KL ball) \\[3pt]
$\tilde{\pi}$, $\tilde{\pi}_I$ && The distributions $\pi$, $\pi_I$ approximated by a bag of Monte Carlo samples\\[3pt]
$\pi^{\sup}_{a,C}$ && The least favourable distribution for action $a$ in the KL ball of radius $C$ centred at $\pi_I$\\[3pt]
$\psi^{\sup}_a(C)$ && Expected loss of action $a$ under $\pi^{sup}_{a,C}$ \\[3pt]
$[\psi^{\inf}_a(C), \psi^{\sup}_a(C)]$ && Interval of expected loss of  action $a$ in $\Gamma_C$\\[3pt]
$\pi^{\sup}_{(a,a'),C}$ && Least favourable distribution corresponding to regret loss function $L_{(a,a')}(\theta)$\\[3pt]
\bottomrule
\end{tabular}\label{glossary}
\end{table*}

\section{Local Sensitivity Analysis}\label{localsens}
We can look at the derivative of least favourable expected loss for a given action either as a function the ball size $C$ or the tilting parameter $\lambda$. Firstly, differentiating wrt $\lambda$ gives:
\[ \frac{\mathrm{d}}{\mathrm{d}\lambda} \E_{\pi_{a,c(\lambda)}^{\sup}}[L_a] = \mathrm{Var}_{\pi^{\sup}_{a,C(\lambda)}} [L_a(\theta)]
\]
Setting $\lambda$ to 0, we see that the sensitivity of the expected loss estimate is given by the variance of the loss under $\pi_I$. Differentiating now w.r.t. $C$ we need the following (applying the chain rule):
\[  \frac{\mathrm{d}}{\mathrm{d}\lambda} C_{\lambda} =  \E_{\pi^{\sup}_{a,C(\lambda)}}[L_a(\theta)] - \E_{\pi_I} [L_a(\theta)]\]
\begin{proof}
We define $\psi(\lambda) = \E_{\pi^{\sup}_{a,C(\lambda)}}[L_a(\theta)] = \int_{\Theta} L_a(\theta)\pi_I(\theta)e^{\lambda L_a(\theta)}Z_{\lambda}^{-1} \mathrm{d}\theta$

where $Z_{\lambda} = \int_{\Theta} \pi_I(\theta)e^{\lambda L_a(\theta)}\mathrm{d}\theta$ (normalising constant).

\[ \frac{\mathrm{d} \psi}{\mathrm{d}\lambda} = \frac{\mathrm{d}}{\mathrm{d}\lambda} \int_{\Theta} L_a(\theta) \pi^{\sup}_{a,C(\lambda)}(\theta) \mathrm{d}\theta 
 = \int_{\Theta} L_a(\theta)\pi_I(\theta)  \frac{\mathrm{d}}{\mathrm{d}\lambda} \left( e^{\lambda L_a(\theta)} Z^{-1}_{\lambda} \right) \mathrm{d}\theta\]
\[ = \int_{\Theta} L_a(\theta)\pi_I(\theta) \left( \frac{L_a(\theta)e^{\lambda L_a(\theta) } Z_{\lambda} - e^{\lambda L_a(\theta)} \frac{\mathrm{d} Z_{\lambda}}{\mathrm{d}\lambda} }{Z_{\lambda}^2} \right) \mathrm{d}\theta  \]
\[ = \int_{\Theta} L_a(\theta)^2 \pi_I(\theta) e^{\lambda L_a(\theta)} Z_{\lambda}^{-1} \mathrm{d}\theta - \int_{\Theta} L_a(\theta)\pi_I(\theta)e^{\lambda L_a(\theta)} Z_{\lambda}^{-1} \left( \int_{\Theta} L_a(\theta)\pi_I(\theta)e^{\lambda L_a(\theta)} Z_{\lambda}^{-1} \mathrm{d}\theta  \right) \mathrm{d}\theta \]
\[ = \E_{\pi^{\sup}_{a,C(\lambda)}}[L_a(\theta)^2] - \E_{\pi^{\sup}_{a,C(\lambda)}} [L_a(\theta)]^2  = \mathrm{Var}_{\pi^{\sup}_{a,C(\lambda)}} [L_a(\theta)]\]

For $\lambda > 0$, define the corresponding KL divergence $C_{\lambda}$:

\[ C_{\lambda} := K(\lambda) := \int_{\Theta} \pi_I(\theta) \log \frac{\pi_I(\theta) Z_{\lambda}}{\pi_I(\theta) e^{\lambda L_a(\theta)}} \mathrm{d}\theta  \]

Hence:

\[ \frac{\mathrm{d}K}{\mathrm{d}\lambda} =  \frac{\mathrm{d}}{\mathrm{d}\lambda}\int_{\Theta} \pi_I(\theta)\left(\log Z_{\lambda} - \lambda L_a(\theta) \right)\mathrm{d}\theta =  \frac{\mathrm{d}}{\mathrm{d}\lambda} \log Z_{\lambda} - \int_{\Theta}  \frac{\mathrm{d}}{\mathrm{d}\lambda}\lambda \pi_I(\theta) L_a(\theta) \mathrm{d}\theta \]
\[ = Z_{\lambda}^{-1} \int_{\Theta}L_a(\theta) \pi_I(\theta) e^{\lambda L_a(\theta)} \mathrm{d}\theta - \int_{\Theta} \pi_I(\theta) L_a(\theta)\mathrm{d}\theta = \E_{\pi^{\sup}_{a,C(\lambda)}}[L_a(\theta)] - \E_{\pi_I} [L_a(\theta)]  \]

So the reciprocal derivative is:

\[  \frac{\mathrm{d}}{\mathrm{d}C_{\lambda}}(K^{-1}) = \frac{1}{ \frac{\mathrm{d}K}{\mathrm{d}\lambda}\left( K^{-1}(C_{\lambda}) \right)}  \]
\end{proof}

\section{Reverse KL neighbourhood: $\KL(\pi_I || \pi)$}\label{KLreverse}

The change of neighbourhood from KL($\pi \parallel \pi_I$) to KL($\pi_I \parallel \pi$) results in a non-analytic solution to the local-minimax and maximin distributions. However we can use  numerical methods to compute the minimax optimisation. We consider the numerical solution to $\pi^{\sup}_{a,C} = \arg\sup_{\pi \in \Gamma^{\text{rev}}_C} \E_{\pi}[L_a(\theta)]$, with $\Gamma^{\text{rev}}_C$  now defined herein as $\Gamma^{\text{rev}}_C = \{\pi : \KL(\pi_I \parallel \pi) \le C\}$ for $C \ge 0$.

\paragraph{Numerical approximation:}
Consider a stochastic representation of $\pi_I$ via
\begin{eqnarray}
\label{eq:stochrep}
\tilde{\pi}_I & \equiv & \frac{1}{m} \sum_{i=1}^m \delta_{\theta_i}(\theta) \\
\theta_i & \sim & \pi_I(\theta) \nonumber
\end{eqnarray}
where $\theta_i$ are iid draws from $\pi_I$ and $m \to \infty$. In practice this is often the actual model that statisticians work with, via a ``bag of samples'' Monte Carlo representation of $\pi_I$. We note for non-degenerate functionals $g(\theta)$ of interest $\E_{\tilde{\pi}_I}[g(\theta)] $ converges to $\E_{\pi_I}[g(\theta)]$, as $m \to \infty$. To make the solution tractable in defining a KL neighbourhood around $\pi_I$ we will use the neighbourhood around $\tilde{\pi}_I$. Moreover, in considering the KL divergence between $\pi_I$ and an alternative model $\pi \in \Gamma^{\mathrm{rev}}_C$ we will work with a stochastic approximation to $\pi$ represented as mixtures of the atoms $\{\delta_{\theta_1}, \delta_{\theta_2}, \ldots, \delta_{\theta_m}\}$ in (\ref{eq:stochrep}), 
\begin{eqnarray}
\label{eq:is}
\tilde{\pi} & = & \sum_i w_i \delta_{\theta_i}(\theta)
\end{eqnarray}
for $0 \le w_i \le 1$, $\sum_i w_i=1$, where the $w_i$'s can be interpreted as importance weights $w_i \propto \pi(\theta_i) / \pi_I(\theta_i)$, with $\E_{\tilde{\pi}}[g(\theta)] \to \E_{\pi}[g(\theta)]$, as $m \to \infty$.

The KL divergence between $\pi_I$ and $\pi$ can then be approximated via the KL divergence of their stochastic representations, 
$$
{\KL}(\tilde{\pi}_I, \tilde{\pi}) =  \frac{1}{m} \sum_{i=1}^m \log \frac{1}{m*w_i}.
$$
with KL ball $\tilde{\Gamma}^{\mathrm{rev}}_C$ defined similarly, $\tilde{\Gamma}^{\mathrm{rev}}_C = \{\tilde{\pi} : \KL(\tilde{\pi}_I, \tilde{\pi}) \leq C\}$.

From these definitions, we will now look for the probability measure maximisation 
\begin{equation}
\label{eq:pi_approx}
\tilde{\pi}^{\sup}_{a,C} = \arg \sup_{\tilde{\pi} \in \tilde{\Gamma}^{\mathrm{rev}}_C} \left\{ \E_{\tilde{\pi}}[L_a(\theta)] \right\}
\end{equation}
Given the atomic structure of $\tilde{\pi}$ the maximisation (\ref{eq:pi_approx}) leads to a convex optimisation in the weights,
\begin{align*}
\tilde{\pi}^{\sup}_{a,C}  = & \sum_i w^*_i \delta_{\theta_i}(\theta) \\
w^*  = & \arg \sup_{w} \left\{  \sum_i w_i L_a(\theta_i) :  - \frac{1}{m}\sum_i \log(w_i)  \leq  C + \log m, \sum_i w_i  =  1 \right\}
 \end{align*}
for which standard numerical methods can be applied.

\bibliographystyle{authordate1}
\bibliography{ref}

\end{document}